\documentclass[amsmath,amssymb,notitlepage]{revtex4-1}

\usepackage{latexsym}
\usepackage{amsfonts}
\usepackage{amsthm}
\usepackage{color}
\usepackage{bbm}
\usepackage{dsfont}
\usepackage{graphicx}
\usepackage{textcomp}
\usepackage[nottoc,numbib]{tocbibind}
\usepackage{enumerate}
\usepackage{graphicx,tikz}
\usetikzlibrary{arrows,matrix,calc}
\usepackage[colorlinks=true,linkcolor=blue,citecolor=blue,urlcolor=blue]{hyperref}
\usepackage{mathrsfs} % Calligraphic font

\newtheorem{theorem}{Theorem}
\newtheorem{proposition}[theorem]{Proposition}
\newtheorem{lemma}[theorem]{Lemma}
\newtheorem{corollary}[theorem]{Corollary}
\newtheorem{definition}[theorem]{Definition}

\newcommand{\ket}[1]{|#1\rangle} %ket
\newcommand{\bra}[1]{\langle#1|} %bra
\newcommand{\eg}{\emph{e.g.}~}

\DeclareMathOperator{\Tr}{Tr} %trace symbol

\DeclareMathOperator{\reals}{\mathbb{R}}
\DeclareMathOperator{\comp}{\mathbb{C}} 
\DeclareMathOperator{\iden}{\mathbbm{1}}

\newcommand{\norm}[1]{\left\lVert#1\right\rVert}
\DeclareMathOperator{\sgn}{sgn}
\begin{document}
\title{A stable quantum Darmois-Skitovich theorem}

\author{Javier Cuesta}
\email{j.cuesta@tum.de}
\affiliation{Department of Mathematics, Technische Universit\"at M\"unchen, 85748 Garching, Germany}
\affiliation{Munich Center for Quantum Science and Technology (MCQST), M\"unchen, Germany}

\date{\today}% It is always \today, today,

\vspace{-1cm}\begin{abstract}
\noindent
The Darmois-Skitovich theorem is a simple characterization of the normal distribution 
in terms of the independence of linear forms. We present here a non-commutative version of this theorem 
in the context of Gaussian bosonic states and show that this theorem is stable under small errors in its underlying conditions. An explicit estimate of the stability constants which depend on the physical parameters of the problem is given.

\end{abstract}

\maketitle

\tableofcontents

\section{ Introduction and Summary of results}

Among all characterizations of the normal distribution, the ones concerning the independence of linear forms stand out because of their simplicity. The landmark result of such classical characterizations is due to Darmois \cite{DG53} and Skitovich \cite{SV54}. Their theorem is a generalization to $n-$random variables and arbritrary coefficients of the following fact: if $X,Y$  are independent real-valued random variables with $X+Y$ and $X-Y$ independent, then $X$ and $Y$ are normally distributed with the same variance (see Theorem~\ref{teo:ds}). We will be interested in studying a quantum (read as non-commutative) version of the Darmois-Skitovich theorem, which we now write shortly as DS theorem. In this case, the role of the normal distribution is taken by Gaussian bosonic states: quantum states whose statistics is completely determined by the knowledge of the first and second moments and whose canonical observables obey the bosonic commutation relations.\\

Particularly noteworthy is the fact that the quantum DS theorem has a clear physical realization. Consider an arbitrary product state that passes through a beam-splitter as in figure \ref{fig:operGBS}. If the output state of the beam-splitter is also a product state, then the input states are Gaussian bosonic states with the \textit{same} second moments. This is the content of the quantum DS-theorem. Mathematically, the content of the quantum DS theorem is given on theorem \ref{teo:ids}.\\ 

That there does not exist any two copies of identical non-Gaussian states fulfilling this is by no means trivial, since the action of a beam splitter does not create second-moment cross-correlations for an identical product of quantum states (Gaussian and non-Gaussian). This operational characterization of Gaussian states was already known\cite{MRKT91}, however without a direct reference to the DS theorem. We show this characterization for a general $n-$mode Gaussian bosonic state by means of the DS theorem. This has the advantage of a much clear statistical interpretation and a simpler proof. Additionally, we show that a beam splitter is the only non-trivial linear operation that can have a factorizable output for \textit{all} identical input states (Lemma \ref{teo:mainclass}). The latter places the beam splitter as the basic element for detecting non-Gaussianity.\\

\begin{figure}[h]
\includegraphics[scale=1]{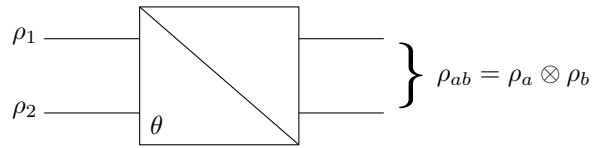}
\caption{Quantum Darmois-Skitovich theorem: let $U_{S}$ be the unitary operator corresponding to the action of a non-trivial beam-splitter transformation. The output state $\rho_{ab}:=U_{S}(\rho_{1}\otimes\rho_{2})U_{S}^{*}$ is a product state if and only
if $\rho_{1}$ and $\rho_{2}$ are Gaussian bosonic states with the same moments.}
\label{fig:operGBS}
\end{figure}

Of course in real life we cannot completely guarantee that two states are totally independent. Therefore it is crucial to study how stable the DS theorem is. This means, how does the conclusion of the quantum DS theorem changes, when we assume that the output state is not exactly a product, but is \textit{approximately close} to a product state. Our main result is a proof of the stability of the quantum DS theorem for quantum states whose all statistical moments in position and momentum, including mixed moments, exist and are finite. Such states are described by the set of Schwartz density operators~\cite{KKW15}. For independent input states whose output from a beam splitter is close in trace norm to a product state, we show that they are close in Hilbert-Schmidt norm to their respective Gaussian counterpart (i.e.~the Gaussian state which has the same first and second moments). Moreover, the corresponding second moments of the input states have to be approximately close as well. A precise mathematical statement of this result is given in theorem \ref{teo:pids}. The robustness of the quantum DS theorem depends on the transmitivity of the beam-splitter, the number of modes and the largest fourth moment of the output state. We give an explicit dependence on these physical parameters in the proof of this theorem.\\

The layout of the paper is as follows. In the next subsection we give the basic definitions and results in continuous-variable quantum information that will be used. Section \ref{mainresultssection} contains the main result. We give a simple proof for the characterization of Gaussian bosonic states using the DS theorem and then proceed to state the stability of the DS theorem. Section \ref{stabilitysubsection} introduces and summarizes the main properties of Schwartz operators. In section \ref{proofofteo:pids} we give a full proof of the stability of the DS theorem. Finally, in section~\ref{technicallemmassubsection} we show some auxiliary lemmas and in the appendix~\ref{sec:appendix} some explicit bounds are calculated in detail. They will be part of the constants appearing in the stability of the DS theorem \ref{teo:pids}.

\subsection{Notation and preliminaries}\label{notationsubsection}
We will be entirely concerned with continuous-variable systems with a discrete number $n$ of
modes. We denote by
\begin{equation}\label{eq:canop}
R:=(Q_{1},P_{1},\ldots,Q_{n},P_{n}),
\end{equation}
the vector of canonical operators for a quantum system and 
$R_{k},k=1,\ldots2n$ its components. Here 
$Q_{l},P_{l},l=1,\ldots n$ act on the $l-$tensor factor of the Fock 
space $\mathcal{H}=\bigotimes_{k=1}^{n}L^{2}(\reals)$ where $L^{2}(\reals)$ denotes 
the space of Lebesgue square integrable functions on $\reals$. 
The \textit{canonical commutation relations} (CCR) are defined  by
\begin{equation}
[R_{k},R_{l}]=i\sigma_{kl},
\label{eq:ccr}
\end{equation}
where $\sigma_{ij}$ are the entries of the symplectic matrix
\begin{equation}
\sigma=\bigoplus_{i=1}^{n}\omega \quad \text{with} \quad \omega:=\begin{pmatrix}{} 0 & 1 \\ -1 & 0 \end{pmatrix}.
\label{eq:sigmambasis}
\end{equation}

We frequently use the shorthand notation $R_{\xi}:=\xi\cdot\sigma R,\qquad \xi\in\reals^{2n}$. The phase-space description of a quantum state $\rho$ is determined by the \textit{characteristic function} $\chi:\reals^{2n}\to\mathbb{C}$ defined by
\begin{equation}\label{eq:defchf}
\chi(\xi):=\Tr[W_{\xi}\rho],
\end{equation}
where $W_{\xi}=e^{i\xi\cdot \sigma R}$ is the so-called Weyl operator. The CCR are encoded in the Weyl relation
 \begin{equation}
 W_{\xi}W_{\eta}=e^{-\frac{i}{2}\xi\cdot\sigma\eta}\ W_{\xi+\eta}, \quad\quad \xi,\eta\in\reals^{2n}.
 \label{eq:weylrel}
 \end{equation}
 The name of characteristic function for the map in Eq.~\eqref{eq:defchf} comes from an analogy with the classical characteristic function which is the Fourier transform of a probability distribution. In fact by taking a fixed direction in phase space we recover the classical characteristic function and from there, we can ``import'' all the known results of the classical world. This indeed, will be used in order to give a simple proof of the characterization of Gaussian bosonic states. The condition for a function $\chi:\reals^{2n}\to\mathbb{C}$ to be a bona-fide quantum characteristic function is the property of sigma-positiveness. For clarity we state these results and refer the reader to Ref.
 \onlinecite[section 5.4]{Hol82} for a proof.
 \begin{theorem}[Quantum Bochner-Khinchin]
For $\chi:\reals^{2n}\to\mathbb{C}$ to be a characteristic function of a quantum state, the following conditions are necessary and sufficient:
\begin{enumerate}
 \item $\chi(0)=1$ and $\chi$ is continuous at $\xi=0$,
 \item $\chi$ is $\sigma$-positive definite, i.e. for any $m\in\mathbb{N}$, 
 any set $\{\xi_{1},\xi_{2},\ldots,\xi_{m}\}$ of vectors in $\reals^{2n}$, and any set $\{c_{1},c_{2},\ldots,c_{m}\}$ of complex numbers
\begin{equation}
\sum_{k,l=1}^{m}c_{k}\overline{c_{l}}\ \chi(\xi_{k}-\xi_{l})\ e^{\frac{i}{2}\xi_{k}\cdot\sigma\xi_{l}}\geq0
\label{eq:sigmapos}
\end{equation}
\end{enumerate}
\label{teo:qbk}
\end{theorem}
 
  \begin{corollary}[Classical Marginals] \label{cor:marginals}
 Let $\chi(\xi)$ be the characteristic function of a quantum state. Then for every fixed $\xi\in\reals^{2n}$ the
 function
 \begin{equation*}
 \reals\ni t\mapsto\chi(t\xi),
 \end{equation*}
 is a classical characteristic function, i.e. the Fourier transform of a classical probability distribution. 
 \end{corollary}

 As in the classical case, the characteristic function is a moment generating function. The \textit{displacement vector} is defined by the entries $d_{k}:=\Tr[\rho R_{k}]$ and we say that the state is \textit{centered} if $d=0$. The \textit{covariance matrix} (CM) is defined by the matrix entries $\varGamma_{kl}=\Tr[\rho\{R_{k}-d_{k},R_{l}-d_{l}\}]$. In order that $\varGamma$ corresponds to a genuine quantum CM the CCR impose the further condition~\cite{SMD94,SSM87} $\varGamma+i\sigma\geq0$, which is nothing but the uncertainty principle expressed in a coordinate-free form.\\

 A \textit{Gaussian bosonic state} is defined as a state with a Gaussian characteristic function
\begin{equation}
\chi(\xi)=\text{exp}[-\frac{\xi\cdot\varGamma\xi}{4}+i\xi\cdot d].
\end{equation}
We write $\chi_{\rho}$ to emphasize that $\chi$ is the characteristic function of the state $\rho$. We denote by $M(2n,\reals)$ and $Sp(4n,\reals)$ the set of $2n\times2n$ matrices with real entries and the group of $4n\times4n$ symplectic matrices with real entries, respectively. The latter is defined as the group of matrices $S\in M(2n,\reals)$ such that $S\sigma S^{T}=\sigma$.\\

Unitary Gaussian operations, i.e. unitary evolutions coming from quadratic Hamiltonians in $P$ and $Q$, are described by symplectic transformations \cite{ADMS95}. 
These operations have the property that
\begin{equation}
 \chi_{(U_{S}\rho U^{*}_{S})}(\xi)=\chi_{\rho}(S^{T}\xi),
 \label{eq:cov}
\end{equation}

where $U_{S}$ is a unitary operation associated to the symplectic transformation $S$ (strictly speaking $U_{S}$ is determined up to a phase, however this ambiguity disappears in the conjugation $U_{S}\cdot U_{S}^{*}$). The unitary evolution of a Gaussian state is completely determined by the new displacement vector $d'=Sd$ and CM,
$\varGamma'=S\varGamma S^{T}$. \\
A one mode \textit{non-trivial beam splitter} transformation is the one 
corresponding to the symplectic transformation
\begin{equation*}
S=\begin{pmatrix}{} \cos\theta\iden_{2} & \sin\theta\iden_{2} \\  -\sin\theta\iden_{2} & \cos\theta\iden_{2} \end{pmatrix}, \qquad \theta\neq m\pi/2,\quad m=0,1,2,\ldots
\label{eq:bseq}
\end{equation*}
It corresponds to a unitary evolution where the Hamiltonian is 
\begin{equation*}
H=\frac{\theta}{4}(a_{1}^{*}a_{2}+a_{2}^{*}a_{1}),
\end{equation*}
with  $a_{j}=(Q_{j}+iP_{j})/\sqrt{2},a^{*}_{j}=(Q_{j}-iP_{j})/\sqrt{2}$, $j=1,2$ the creation and annihilation operators. 
A \textit{local transformation} acts in each separated mode and corresponds therefore to transformations that can be written as $S=\bigoplus_{k=1}^{n}S_{k}$. 
In the context of quantum optics, example of local transformations are phase-shifts and one-mode squeezing transformations.\\

The \textit{Wigner phase space distribution} is defined to be the (symplectic) Fourier transform of the characteristic function

\begin{equation}
\mathcal{W}(\eta)=\frac{1}{(2\pi)^{n}}\int e^{i\eta\cdot\sigma\xi}\;\chi(\xi)\;d\xi.
\end{equation}
Its importance lies on the fact that all one-dimensional marginals are (due to corollary \ref{cor:marginals}) positive distributions in phase space, which can be asociated to the usual probability distributions for example on position and momentum of a state $\rho$. \\

We write $A^{*}$ for the adjoint operator of $A$ and $\norm{\cdot}$ for the uniform norm. The trace norm is defined as $\norm{A}_{1}=\Tr\sqrt{A^{*}A}$ and the Hilbert-Schmidt (HS) norm $\norm{A}_{2}=\left(\Tr[A^{*}A]\right)^{1/2}$. We have the order $\norm{\cdot}\leq\norm{\cdot}_{2}\leq\norm{\cdot}_{1}$. These norms are in fact unitarily invariant and $\norm{A^{*}}_{p}=\norm{A}_{p}$ for $p=1,2$. The usual norm in $L^{2}(\reals^{2n})$ will be denoted by $\norm{\cdot}_{L^{2}(\reals^{2n})}$. We use sometimes the Dirac notation for a vector $\ket{\phi}\in\mathcal{H}$ and the inner product notation $\bra{\phi}\varphi\rangle$. The commutator and anticommutator are written as $[\cdot,\cdot]$ and $\{\cdot,\cdot\}$ respectively. The space of bounded operators on the Hilbert space $\mathcal{H}$ is denoted by $\mathfrak{B}(\mathcal{H})$.\\

The inverse relation of Eq.~\eqref{eq:defchf} is called the \textit{Weyl transform}

\begin{equation}
T=\frac{1}{(2\pi)^{n}}\int \Tr[W_{\xi}T]W_{-\xi}d\xi,
\end{equation}
where the integral converges weakly for any Hilbert-Schmidt operator $T$. This is a consequence of the \textit{quantum Parseval theorem}\cite{Hol82} which due to its importance we state here.

\begin{theorem}[Quantum Parseval relation]\label{th:quantumparseval} Let $\{W_{\xi}\}$ be a strongly continuous and irreducible Weyl systems acting on the Hilbert space $\mathcal{H}$ with respective phase space $X\simeq\reals^{2n}\ni\xi$. Then $T\mapsto \Tr[W_{\xi}T]$ extends uniquely to an isometric map from the Hilbert space of Hilbert-Schmidt class operators on $\mathcal{H}$ onto $L^{2}(X)$, such that 
\begin{equation}
\Tr T^{*}_{1}T_{2}=\frac{1}{(2\pi)^{n}}\int \overline{\Tr[W_{\xi}T_{1}]}\Tr[W_{\xi}T_{2}]d\xi.
\end{equation}
\end{theorem}

This theorem also implies that Eq.~\eqref{eq:defchf} is also valid for $T$ Hilbert-Schmidt. The map $\xi\mapsto\Tr W_{\xi}T$ is called the \textit{inverse Weyl transform} of $T$; being the characteristic function the special case $T$ a density operator.\\
We will be using repeatedly the following trace inequalities. If $B$ is a bounded operator and $T$ a trace-class operator, then a particular case of H\"older's inequality states
\begin{equation*}
\Tr BT\leq \norm{B}\norm{T}_{1}.
\end{equation*}
Let $T_{1},T_{2}$ be two Hilbert-Schmidt operators. The trace operator version of the Cauchy-Schwarz inequality is
\begin{equation*}
\Tr T_{1}T_{2} \leq \norm{T_{1}T_{2}}_{1}\leq \norm{T_{1}}_{2}\norm{T_{2}}_{2}.
\end{equation*}

 \section{Main result}\label{mainresultssection}

In the next subsections we present the quantum version of the DS theorem and our main stability result. The detailed proof of the stability of the DS theorem is presented in section \ref{proofofteo:pids}. 

\subsection{Quantum Darmois-Skitovich theorem}\label{mainresults1}

We are interested in a quantum analogue of the following theorem.

\begin{theorem}[\textit{Darmois-Skitovich}]
Let $X_{1},\ldots,X_{n}$ ($n\geq2$) be independent random variables and $a_{1},\ldots,a_{n},b_{1},\ldots,b_{n}\in\reals\backslash\{0\}$. 
If the two linear forms
\begin{equation}\label{eq:linearforms}
 Y_{1}=\sum_{i}a_{i}X_{i}\quad \text{and}\quad
Y_{2}=\sum_{i}b_{i}X_{i}\quad \text{are independent,}
\end{equation}
then $X_{i}$ is normally distributed.
\label{teo:ds}
\end{theorem}
Different proofs and the history of the classical DS theorem can be found in p. 78 in Ref.~\onlinecite{feller-vol-2} and in Ref.~\onlinecite{KLR73}.
Our setup for the quantum version is the following.  We consider two $n-$mode quantum states $\rho_{1},\rho_{2}\in\mathfrak{B}(\mathcal{H})$ with respective canonical operators 
\begin{equation}\label{eq:notationR1R2}
\begin{split}
R_{1}&=(Q_{1},P_{1},\ldots,Q_{n},P_{n}) \\
R_{2}&=(Q_{n+1},P_{n+1},\ldots,Q_{2n},P_{2n})
\end{split}
\end{equation}
and write $R=(R_{1},R_{2})$. We assume that $\rho_{1}$ and $\rho_{2}$ are independent so that their state in $\mathfrak{B}(\mathcal{H}\otimes\mathcal{H})$ is a product state $\rho_{1}\otimes\rho_{2}$. We refer to $\rho_{1}$ and $\rho_{2}$ as input states.\\

The action of producing linear forms of random variables can be mimiced by (Gaussian) unitary evolutions $U_{S}$. These unitary evolutions are generated by Hamiltonians that are quadratic expressions in the canonical operators. Moreover\cite{ADMS95}, the unitary evolution $U_{S}\in\mathfrak{B}(\mathcal{H}\otimes\mathcal{H})$ is associated with a symplectic transformation $S\in Sp(4n,\reals)$. In other words, the linear trasformation
\begin{equation}\label{eq:condition_S}
R\mapsto SR \qquad \qquad \qquad S\in Sp(4n,\reals),
\end{equation}
 corresponds to a unitary evolution $\rho\mapsto U_{S}\rho U^{*}_{S}$.\\

In order to obtain an analogue of Eq.~\eqref{eq:linearforms}, we classify the set of unitaries $U_{S}$ which produce a bipartite independent output, i.e. such that

\begin{equation}\label{eq:condition}
U_{S}(\rho_{1}\otimes\rho_{2}) U^{*}_{S}=\rho_{a}\otimes\rho_{b},
\end{equation}
where $\rho_{a},\rho_{b}\in\mathfrak{B}(\mathcal{H})$ are $n-$mode quantum states. This is equivalent to classifying the respective set of symplectic transformation for which Eq.~\eqref{eq:condition} holds. If we are to expect that $U_{S}$ preserves independence, the transformation $S$ should at least preserve uncorrelated inputs (a generally weaker condition than independence which only deals with the second moments). Furthermore, acting locally on each state or swapping them are trivial operations which preserve independence for arbitrary states. So we need to consider other operations in order to obtain a meaningful statement for the quantum DS theorem. The following two lemmas show in fact that there is only one non-trivial symplectic transformation for our setup.

\begin{lemma}
Let $S\in Sp(4n,\mathbb{R})$ be such that
\begin{equation*}
 S\begin{pmatrix}{} \varGamma_{1} & 0 \\ 0 & \varGamma_{2} \end{pmatrix}S^{T}=\begin{pmatrix}{} \cdot & 0 \\ 0 & * \end{pmatrix} \quad \text{for all CM}\quad\varGamma_{1},\varGamma_{2}\in\mathbb{R}^{2n\times 2n}.
 \label{eq:symp2}
 \end{equation*}
 Here $*,\cdot$ denote any $\text{CM}\in\mathbb{R}^{2n\times 2n}$.
Then $S$ is either of the form $\begin{pmatrix}{} A & 0 \\ 0 & D \end{pmatrix}$ or $\begin{pmatrix}{} 0 & B \\ C & 0 \end{pmatrix}$ 
with $A,B,C,D\in Sp(2n,\reals)$.
%\quad \text{and}\quad  \alpha\in\mathbb{R}\cup\{\pm\infty\}.
\label{teo:class2}
\end{lemma}

%%%%%%%%%%%%%%%%%%%%%%%%% Block preserving matrices %%%%%%%%%%%%%%%%%%%%%%%%%%%%%%%%%%%%%%%%%%%%%%%%%%%%%%%%%5
If we consider identical inputs we obtain:

\begin{lemma}
Let $S\in Sp(4n,\mathbb{R})$ be such that
\begin{equation}
 S\begin{pmatrix}{} \varGamma & 0 \\ 0 & \varGamma \end{pmatrix}S^{T}=\begin{pmatrix}{} \cdot & 0 \\ 0 & * \end{pmatrix} \quad \text{for all CM}\quad\varGamma\in\mathbb{R}^{2n\times 2n}.
 \label{eq:symp}
 \end{equation}
 Here $*,\cdot$ denote any $\text{CM}\in\mathbb{R}^{2n\times 2n}$.
Then 
\begin{equation}
\begin{split}
S&=
\begin{pmatrix}{} X & 0 \\  0 & Y \end{pmatrix}\frac{1}{\sqrt{1+\alpha^{2}}}
\begin{pmatrix}{} \iden_{2n} & \alpha\iden_{2n}\\  -\alpha\iden_{2n} & \iden_{2n} \end{pmatrix},\\
&=
\begin{pmatrix}{} 0 & X \\  Y & 0 \end{pmatrix}\frac{1}{\sqrt{1+\gamma^{2}}}
\begin{pmatrix}{} \iden_{2n} & -\gamma\iden_{2n}\\  \gamma\iden_{2n} & \iden_{2n} \end{pmatrix},
\label{eq:main}
\end{split}
\end{equation}
where $X,Y\in Sp(2n,\mathbb{R})$ and $\alpha,\gamma\in\mathbb{R}\cup\{\pm\infty\}$. 
%\quad \text{and}\quad  \alpha\in\mathbb{R}\cup\{\pm\infty\}.
\label{teo:mainclass}
\end{lemma}

Thus the only non-trivial linear transformation in Eq.~\eqref{eq:condition_S} is of the form of Eq.~\eqref{eq:main}. We discard the trivial operations and set 
$\alpha=\tan\theta$ in Eq.~\eqref{eq:main} 

\begin{equation}
 S_{\theta}=\begin{pmatrix}{} \cos\theta\iden_{2n} & -\sin\theta\iden_{2n}\\  \sin\theta\iden_{2n} & \cos\theta\iden_{2n} \end{pmatrix}, \qquad \theta\neq m\pi/2,\quad m=0,1,2,\ldots,
\label{eq:bstheta}
\end{equation}
which is the symplectic transformation associated to a ($n-$mode) beam splitter operation. We refer to the latter operation as a \textit{non-trivial beam splitter} transformation. \\

Although for every covariance matrix $\varGamma$, $S_{\theta}(\varGamma\oplus\varGamma)S_{\theta}^{T}=(\varGamma\oplus\varGamma)$, it turns out from the DS theorem that there does not exist a non-Gaussian state $\rho$ such that $U_{S_{\theta}}(\rho\otimes\rho)U^{*}_{S_{\theta}}=\rho\otimes\rho.$ See figure~\ref{fig:operGBS}.

\begin{theorem}[\textsc{Quantum Darmois-Skitovich}]
Let $U_{S}$ be the unitary operation corresponding to a non-trivial beam spliter Eq. \ref{eq:bstheta}. Consider the state $\rho_{ab}=U_{S}(\rho_{1}\otimes\rho_{2})U_{S}^{*}$ 
obtained after the unitary evolution of an arbitrary product state. If the output state is a product state $\rho_{ab}=\rho_{a}\otimes\rho_{b}$, then 
$\rho_{1}$ and $\rho_{2}$ are Gaussian bosonic states with the same CM but not necessarily same displacement vector.
\label{teo:ids}
\end{theorem}

Due to the 1--1 correspondence between quantum states and their characteristic function we have the following consequence.
\begin{corollary}
 Let $\chi_{1}$ and $\chi_{2}$ be the characteristic function of the quantum states $\rho_{1}$ and $\rho_{2}$ respectively,
 which have finite second moments. If for a fixed $\theta\neq m\pi/2,\quad m=0,1,2,\ldots$ the characteristic functions satisfy the functional equation
\begin{equation}
\chi_{1}(\cos\theta\xi_{1}+\sin\theta\xi_{2})\chi_{2}(\cos\theta\xi_{2}-\sin\theta\xi_{1})=\chi_{1}(\cos\theta\xi_{1})\chi_{1}(\sin\theta\xi_{2})\chi_{2}(\cos\theta\xi_{2})\chi_{2}(-\sin\theta\xi_{1}),
\label{eq:fegs}
\end{equation}
for all $\xi_{1},\xi_{2}\in\mathbb{R}^{2n}$. then $\rho_{1}$ and $\rho_{2}$ are Gaussian bosonic states with the same CM but not necessarily same displacement vector.
\end{corollary}

An inmmediate proof of the quantum DS theorem can be obtained from its corresponding classical result and corollary \ref{cor:marginals}. We 
recall that the latter corollary tells us that we always obtain a positive Wigner function (a classical probability distribution) from the quantum characteristic function
whenever we move through a fixed direction in phase space.

\begin{proof}[Proof Theorem \ref{teo:ids}]
 Using Eq.~\eqref{eq:cov} the evolution of the input states can be expressed in terms of characteristic functions as Eq.~\eqref{eq:fegs}. We fixed the direction 
 $\xi_{1}=\xi_{2}=\xi$ and parametrize $\xi_{1}=t\xi$, $\xi_{2}=s\xi$ with $t,s\in\reals$. Moreover, we introduce the classical characteristic functions 
 $\chi_{j}(u):=\chi_{j}(u\xi), j=1,2$ with $u\in\reals$ so that Eq.~\eqref{eq:fegs} reads
 \begin{equation*}
  \chi_{1}(\cos\theta t+\sin\theta s)\chi_{2}(\cos\theta s-\sin\theta t)=\chi_{1}(\cos\theta t)\chi_{1}(\sin\theta s)\chi_{2}(\cos\theta s)\chi_{2}(-\sin\theta t).
 \end{equation*}
This last equation is the functional version of the classical DS theorem (c.f. Eq. 8.7 of section XV.9 in Ref.~\onlinecite{feller-vol-2}). From this classical result it follows that $\chi_{1}$ and $\chi_{2}$ 
are one dimensional Gaussian characteristic functions with the same variance. We can compute the moments by taking derivates of the characteristic function to obtain that 
\begin{equation*}
 \chi_{j}(t\xi)=\exp[-\frac{t^{2}}{4}(\xi\cdot\varGamma\xi)+it\xi\cdot d_{j}]\qquad j=1,2.
\end{equation*}
The result follows with $t=1$.
\end{proof}

%%%%%%%%%%%%%%%%%%%%%%%%%%%%%%%%%%%%%%%%%%%%%%%%%%%%%%%%%%%%%%%%%%%%%%%%%%%%%%%%%%%%%%%%%%%%%%%%%%%%%%%%%%%%%%%%%%%%%%%%%%%%%%%%%%%%%%%%%%%%%%
%%%%%%%%%%%%%%%%%%%%%%%%%%%%%%%%%%%%%%%%%%%%%%%%%%%%%%%%%%%%%%%%%%%%%%%%%%%%%%%%%%%%%%%%%%%%%%%%%%%%%%%%%%%%%%%%%%%%%%%%%%%%%%%%%%%%%%%%%%%%%%
%%%%%%%%%%%%%%%%%%%%%%%%%%%%%%%%%%%%%%%               STABILITY               %%%%%%%%%%%%%%%%%%%%%%%%%%%%%%%%%%%%%%%%%%%%%%%%%%%%%%%%%%%%%%%%
%%%%%%%%%%%%%%%%%%%%%%%%%%%%%%%%%%%%%%%%%%%%%%%%%%%%%%%%%%%%%%%%%%%%%%%%%%%%%%%%%%%%%%%%%%%%%%%%%%%%%%%%%%%%%%%%%%%%%%%%%%%%%%%%%%%%%%%%%%%%%%
%%%%%%%%%%%%%%%%%%%%%%%%%%%%%%%%%%%%%%%%%%%%%%%%%%%%%%%%%%%%%%%%%%%%%%%%%%%%%%%%%%%%%%%%%%%%%%%%%%%%%%%%%%%%%%%%%%%%%%%%%%%%%%%%%%%%%%%%%%%%%%

\subsection{Stability}\label{stabilitysubsection}

In the last section we presented an exact version of the DS theorem which brings naturally an operational characterization of Gaussian bosonic states. There, it is assumed an exact factorizability of the output state. In practice, it is impossible to assure such thesis since there are always errors in the measurements and therefore the validity of the result is not completely clear in real life. Moreover, any practical application can be inmmediatealy rule out if the conclusion is not robust against small changes in the defined conditions. \\

We are interested in finding to which extent the results of theorem \ref{teo:ids} are affected if the main assumption is not exact but approximately satisfied. Before specifying the conditions of the stability of the quantum DS theorem, we briefly comment on the respective classical stability problem\cite{Lu77}. \\

The stability of the classical DS theorem is due to Yu. R. Gabovich \cite{Ga74}. In his work, the word ``approximately" is quantified in terms of the closeness of the cumulative distribution functions of the random variables. It is shown (Theorem 3 in Ref.~\onlinecite{Ga74}) that for approximate independency the considered classical probability distributions are $c(\log\log(1/\varepsilon))^{-1/8}$-close to a normal distribution in the Levy metric. In this estimate there is no explicit value of the constant $c$ and therefore it is not known how it depends on the coefficients of the linear forms, neither on the number of random variables involved. \\

At the level of density operators it was proven\cite{TopD66} that weak operator topology is equivalent to trace-norm topology.
Therefore Gabovich result and corollary \ref{cor:marginals} should imply at least in a qualitative manner that the quantum DS theorem is stable. However, we can obtain a better concrete estimate of the stability of the DS theorem by considering not just the marginals and the classical result, but rather using the entire phase space and the natural restrictions on the quantum characteristic functions. Namely, due to Heisenberg's uncertainty principle there cannot be characteristic functions which are highly concentrated in some region of phase space. This naturally rules out ill-behaved distributions that can appear in the classical case. With no more preambules, we give a precise statement of our result. \\ \\

We say that a quantum state $\rho_{ab}$ is an $\varepsilon-$\textit{approximate product state} if there is a product state $\rho_{a}\otimes\rho_{b}$ such that
 
\begin{equation}
\norm{\rho_{ab}-\rho_{a}\otimes\rho_{b}}_{1}\leq\varepsilon,
\label{eq:appsep}
\end{equation}
Suppose that two independent states evolve according to the action of a beam splitter, but this time the output state is an $\varepsilon-$approximate product state (see Fig.~\ref{fig:inesGBS}). 

\begin{figure}[h]
 \includegraphics[scale=1]{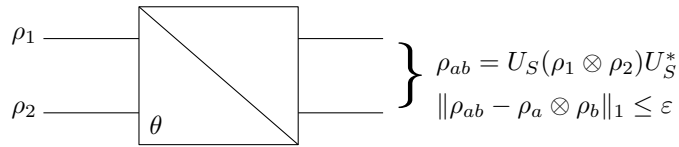}
 \caption{Stability of the Darmois-Skitovich theorem. In this case we consider the output state to be an $\varepsilon-$approximate product state.}
 \label{fig:inesGBS}
 \end{figure}

The following theorem describes the robustness of the quantum DS theorem.

\begin{theorem}[Stability DS]\label{teo:pids}
Let $U_{S}$ be the unitary operation corresponding to a non-trivial beam spliter charaterized in Eq.~\eqref{eq:bstheta}
and $\rho_{1},\rho_{2}$ density operators of two $n$-mode systems with finite moments of all orders and whose CMs are $\varGamma_{1}$ and $\varGamma_{2}$. Consider the output state $\rho_{ab}=U_{S}(\rho_{1}\otimes\rho_{2})U_{S}^{*}$ and define by $\rho'_{1},\rho_2'$ Gaussian bosonic states that have the same CMs and displacement vectors as $\rho_{1},\rho_{2}$.  
If for sufficiently small $\varepsilon\in(0,1)$ the output state $\rho_{ab}$ is $\varepsilon$-close to a product state in trace norm, then 

\begin{equation}\label{eq:boundstabilDS}
 \norm{\rho_{j}-\rho'_{j}}_{2}\leq c_{1}\varepsilon^{1/3}+\frac{c_{2}}{\sqrt{\log(1/\varepsilon)}}, \qquad j=1,2,
\end{equation}
\begin{equation}
 \norm{\varGamma_{1}-\varGamma_{2}}_{2}\leq c_{3}\varepsilon^{1/2},
\end{equation}
where the constants $c_{1},c_{2}$ and $c_{3}$ (which are made explicit in the proof) depend on the transmitivity $\theta$ of the beam-splitter, the number of modes $n$ and the second and fourth absolute moments of the output state $\rho_{ab}$. 
\end{theorem}

Note that under the conditions of theorem~\ref{teo:pids} we also have that the respective Wigner functions are close with the bound of Eq.~\eqref{eq:boundstabilDS}. Furthermore, the stability result could be in principle be extended to a larger set of quantum states since Schwartz operators are dense in the set of trace class operators (see Lemma 2.5 in \cite{KKW15}).

%%%%%%%%%%%%%%%%%%%%%%%%%%%%%%%%%%%%%%%%%%%%%%%%%%%%%%%%%%%%%%%%%%%%%%%%%%%%%%%%%%%%%%%%%%%%%%%%%%%%%%%%%%%%%%%%%%%%%%%%%%%%%%%%%%%%%%%%%%%%%%
%%%%%%%%%%%%%%%%%%%%%%%%%%%%%%%%%%%%%%%%%%%%%%%%%%%%%%%%%%%%%%%%%%%%%%%%%%%%%%%%%%%%%%%%%%%%%%%%%%%%%%%%%%%%%%%%%%%%%%%%%%%%%%%%%%%%%%%%%%%%%%
{\renewcommand\addcontentsline[3]{} \subsubsection*{Schwartz operators}} \label{stabilitysubsection}
%\subsubsection*{Schwartz operators}\label{stabilitysubsection}
 
We begin this section by reviewing some important facts about Schwartz operators which are the non-commutative analogue of Schwartz functions. The latter are infinite differentiable functions whose derivatives decay faster than any polynomial at infinity. The introduction of Schwartz operators allows us to handle differentiability and boundness problems in an elegant manner  and it is for this reason that they play an important technical role in this paper. This class of operators was
first introduced in Ref. \onlinecite{KKW15}, and the reader may find there a detailed exposition. \\

In our proof of the stability of the DS theorem, we deal with terms of the form $\Tr[R_{k}R_{l}\rho R_{s}R_{r}]$ which are a priori not necessarily well-defined on any dense domain of $\rho$; it may happen that $\rho$ maps outside the domain of $R_{l}$. This is a common issue when dealing with unbounded operators. There is an inmense advantage when working with Schwartz operators since many regular properties for bounded operators become available for unbounded ones (\eg ``cycling under the trace" property holds). Indeed, Theorem~\ref{theo:schwprop} below plays a decisive role in the calculation of the constants appearing in our proof of the stability of DS theorem.\\

\begin{definition}[Schwartz Operators]\label{def:schwarzop} An operator $T\in\mathcal{B}(\mathcal{H})$ is called a Schwartz operator if
\begin{equation*}
 \norm{P^{\alpha}Q^{\beta}T P^{\alpha'}Q^{\beta'}}_{1}<\infty \qquad \quad \text{for all }\alpha,\alpha',\beta,\beta'\in I_{n},
\end{equation*}
where $I_{n}:=\{\alpha=(\alpha_{1},\ldots,\alpha_{n})|\alpha_{i}\in\mathbb{N}\cup\{0\}\text{ for all }i=1,\ldots,n\}$ is the set of multi-indices, and 
\begin{equation}
 Q^{\alpha}=Q^{\alpha_{1}}_{1}\cdots Q^{\alpha_{n}}_{n}, \qquad \quad  P^{\alpha}=P^{\alpha_{1}}_{1}\cdots P^{\alpha_{n}}_{n}.
\end{equation}
The set of all Schwartz operators will be denoted by $\mathfrak{S}(\mathcal{H})$.
\end{definition}

So for Schwartz-density operators all the statistical moments in $Q$ and $P$ exist and are finite. We denote by 
\begin{equation*}
\mathscr{S}(\mathcal{H}):=\{\rho\in \mathfrak{S}(\mathcal{H})\ |\ \rho \text{ is a density operator}\ \}
\end{equation*}
the space of density operators which are also Schwartz operators. 
 For $\rho\in\mathscr{S}(\mathcal{H})$ we have the following neat characterization

\begin{proposition}\label{prop:schwfunc}
Let $T$ be a Hilbert-Schmidt operator. Then T is a Schwartz operator if and only if the respective Weyl transform is a Schwartz function.
\end{proposition}
% The proof of this follows the characterization of schwartz operators in terms of the HS-Kernel and from Plancherel Theorem ( see Wong. Theo. % 1.13)

\begin{corollary}\label{cor:schwtrace}
A density operator $\rho$ is a Schwartz operator if and only if its characteristic function $\chi$, or Wigner function $\mathcal{W}$, is a Schwartz function. Moreover, the partial trace of a Schwartz-density operator is a Schwartz operator.
\end{corollary}

The following Theorem contains the basic properties of Schwartz operators that we use.

\begin{theorem}\label{theo:schwprop}
Let $\mathcal{H}=L^{2}(\reals^{2n})$ and $T\in\mathcal{B}(\mathcal{H})$. Then
 \begin{enumerate}[(i)]
  \item Let $f$ be a polynomial function on the entries of the vector $R=(Q_{1},P_{1},\ldots,Q_{n},P_{n})$ and $\{W_{\xi}\}$ a Weyl system. If $T\in\mathfrak{S}(\mathcal{H})$, then $\Tr[f(R)T]=\Tr[Tf(R)]$. Moreover, $f(R)T\in\mathfrak{S}(\mathcal{H})$ and $\Tr[W_{\xi}f(R)T]=\Tr[f(R)TW_{\xi}]=\Tr[TW_{\xi}f(R)]$.
  \item If $T\in\mathfrak{S}(\mathcal{H})$, then $T$ is trace-class. 
  \item If $T\in\mathfrak{S}(\mathcal{H})$, then $|T|\in\mathfrak{S}(\mathcal{H})$.
  \item If $0<T\in\mathfrak{S}(\mathcal{H})$, then $\sqrt{T}\in\mathfrak{S}(\mathcal{H})$.
  \item If $T\in\mathfrak{S}(\mathcal{H})$ then $T^{*}\in\mathfrak{S}(\mathcal{H})$.
 
 \end{enumerate}
\end{theorem}

% \begin{proof} All the proofs are in Ref. \onlinecite{KKW15}:
%   (i) Lemma 3.6 and the comment thereafter. Also From Prop. 3.13 (b) as Q,P leave invariant the Schwartz space. For the bounded cyclicity we use Lemma 3.9 with the usual property of cyclicity for trace-class operators (ii) Prop. 3.15, (iii) Corollary 3.16, (iv) Corollary 3.20 and Remark 3.22  (v) Lemma 3,5
% \end{proof}

For Schwartz-density operators we can write explicit formulas for the gradient and Hessian of the characteristic function in terms of a trace:
Consider $\{\xi_{k}\}_{k=1}^{2n}$ any basis in $\reals^{2n}$, then the gradient of $\chi(\xi)$, denoted by $\nabla\chi(\xi)$, is defined by the entries
\begin{equation*}
 \frac{\partial\chi(\xi)}{\partial\xi_{k}}=\frac{d}{dt}\chi(\xi+t\xi_{k})\Bigr|_{t=0}.
\end{equation*}
 
\begin{lemma}[Gradient of the Weyl Operator]\label{l:qcfmoments2}
Let $T$ be a Schwartz operator and $\nabla_{\eta}:=\eta\cdot\nabla$. Then the following identities hold
\begin{align}
& \chi_{R_{\xi_{k}}T}(\xi)=\left(\frac{1}{2}\xi_{k}\cdot\sigma\xi-i\frac{\partial}{\partial\xi_{k}} \right)\chi_{T}(\xi) \label{eq:dev1}\\
& \chi_{TR_{\xi_{k}}}(\xi)=\left(-\frac{1}{2}\xi_{k}\cdot\sigma\xi-i\frac{\partial}{\partial\xi_{k}} \right)\chi_{T}(\xi) \label{eq:dev2}\\
& (\nabla_{\eta}\chi_{T})(\xi)=\frac{i}{2}\Tr(\{W_{\xi},R_{\eta}\}T)=\frac{i}{2}\Tr(W_{\xi}\{R_{\eta},T\}) \label{eq:gradient2} \\
&(\xi_{k}\cdot\sigma\xi)\chi_{T}(\xi)=\Tr([W_{\xi},R_{\eta}]T)=\Tr(W_{\xi}[R_{\eta},T])  \label{eq:commutatorweyl}
\end{align}
\end{lemma} 

\begin{proof}[Proof of lemma \ref{l:qcfmoments2}]
First note that Eq.~\eqref{eq:gradient2} and Eq.~\eqref{eq:commutatorweyl} follow from adding and substracting Eq.~\eqref{eq:dev1} and Eq.~\eqref{eq:dev2} together with Theorem \ref{theo:schwprop}$(i)$. We show first that 

\begin{equation}\label{eq:dev0}
\frac{d}{dt}(\Tr W_{t\eta}T)\Bigr|_{t=0}=i\Tr R_{\eta}T.
\end{equation}
Since $T$ is trace-class (Theorem \ref{theo:schwprop}$(ii)$) we can decompose $T=T_{1}+iT_{2}$ with $T_{1},T_{2}$ self-adjoint trace-class operators. Moreover, we can write $T_{1},T_{2}$ as a fnite linear combination of positive, trace-class operators and thus from the linearity of the trace we can assume without lost of generality that $T$ is a positive, trace-class operator. From the spectral decomposition $R_{\eta}=\int xdE(x)$ and the functional calculus we obtain

\begin{align*}
\left|\Tr \left(\frac{W_{t\eta}-\iden}{it}-R_{\eta}\right)T \right|&=\left| \int\left(\frac{e^{itx}-1}{it}-x\right)\Tr dE(x)T \right|, \\
&\leq\int \left|\frac{e^{itx}-1}{it}-x\right|\Tr dE(x)T. \\ 
\end{align*}

Using $\left|\frac{e^{itx}-1}{it}-x\right|\leq 2|x|$, the Cauchy-Schwarz inequality and Theorem \ref{theo:schwprop}(iv) 

\begin{align*}
\int \left|\frac{e^{itx}-1}{it}-x\right|\Tr dE(x)T &\leq  2\Tr |R_{\eta}|T, \\ 
&\leq 2\norm{R^{2}_{\eta}\sqrt{T}}_{2}\norm{\sqrt{T}}_{2}<\infty.  
\end{align*}
Hence from the dominated convergence theorem we proved what is required. Now we proceed to prove Eq.~\eqref{eq:dev1}. From Theorem \ref{theo:schwprop}$(i,ii)$ we have that $R_{\xi_{k}}T$ is trace-class and therefore the Weyl transform exists. Moreover, $\chi_{R_{\xi_{k}}T}(\xi)$ is Schwartz, hence continuous, as it is the Weyl transform of a Schwartz operator (Proposition \ref{prop:schwfunc}). So we just need to verify the relation of  Eq.~\eqref{eq:dev1} as Eq.~\eqref{eq:dev2} is similar. This follows directly from Eq.~\eqref{eq:dev0} and \eqref{eq:weylrel} 

\begin{align*}
\chi_{R_{\xi_{k}}T}(\xi)&=-i\frac{d}{dt}\Tr W_{\xi}W_{t\xi_{k}}T\Bigr|_{t=0},\\
&=-i\frac{d}{dt}\left(e^{it\xi_{k}\cdot\sigma\xi/2}\chi(\xi+t\xi_{k})\right)\Bigr|_{t=0}.
\end{align*}

\end{proof} 
 
We remark that since $T$ is Schwartz, higher order derivatives of $\chi_{T}(\xi)$ can be written explicitely by using theorem \ref{theo:schwprop}$(i)$ and Eq.~\eqref{eq:gradient2}.  For instance, the Hessian of $\chi_{\rho}(\xi)$ where $\rho\in\mathscr{S}(\mathcal{H})$ has entries given by
 
\begin{equation}
 \frac{\partial^{2}\chi(\xi)}{\partial\xi_{k}\partial\xi_{l}}=-\frac{1}{4}\Tr[W_{\xi}\{R_{\xi_{k}},\{R_{\xi_{l}},\rho\}\}].
\label{eq:jacob}
 \end{equation} 

In particular, if the state $\rho$ has covariance matrix $\varGamma$ and displacement vector $d$
\begin{equation}
\nabla\chi(0)=i\sigma d\qquad \textnormal{and}	\qquad\textnormal{(Hessian }\chi)(0)=-\sigma\left(\frac{\varGamma}{2}+dd^{T}\right)\sigma^{T}. \label{eq:identitiesGH}
\end{equation}

% 
% 
%%%%%%%%%%%%%%%%%%%%%%%%%%%%%%%%%%%%%%%%%%%%%%%%%%%%%%%%%%%%%%%%%%%%%%%%%%%%%%%%%%%%%%%%%%%%%%%%%%%%%%%%%%%%%%%%%%%%%%
%%%%%%%%%%%%%%%%%%%%%%%%%%%%%%%%%%%%%%%%%%%%%%%%%%%%%%%%%%%%%%%%%%%%%%%%%%%%%%%%%%%%%%%%%%%%%%%%%%%%%%%%%%%%%%%%%%%%%%
%%%%%%%%%%%%%%%%%%%%%%%%%%%%%% Proofs Section %%%%%%%%%%%%%%%%%%%%%%%%%%%%%%%%%%%%%%%%%%%%%%%%%%%%%%%%%%%%%%%%%%%%%%%%
%%%%%%%%%%%%%%%%%%%%%%%%%%%%%%%%%%%%%%%%%%%%%%%%%%%%%%%%%%%%%%%%%%%%%%%%%%%%%%%%%%%%%%%%%%%%%%%%%%%%%%%%%%%%%%%%%%%%%%
%%%%%%%%%%%%%%%%%%%%%%%%%%%%%%%%%%%%%%%%%%%%%%%%%%%%%%%%%%%%%%%%%%%%%%%%%%%%%%%%%%%%%%%%%%%%%%%%%%%%%%%%%%%%%%%%%%%%%%
%%%%%%%%%%%%%%%%%%%%%%%%%%%%%%%%%%%%%%%%%%%%%%%%%%%%%%%%%%%%%%%%%%%%%%%%%%%%%%%%%%%%%%%%%%%%%%%%%%%%%%%%%%%%%%%%%%%%%%
%%%%%%%%%%%%%%%%%%%%%%%%%%%%%%%%%%%%%%%%%%%%%%%%%%%%%%%%%%%%%%%%%%%%%%%%%%%%%%%%%%%%%%%%%%%%%%%%%%%%%%%%%%%%%%%%%%%%%%
%%%%%%%%%%%%%%%%%%%%%%%%%%%%%%%%%%%%%%%%%%%%%%%%%%%%%%%%%%%%%%%%%%%%%%%%%%%%%%%%%%%%%%%%%%%%%%%%%%%%%%%%%%%%%%%%%%%%%%

\section{Proofs}\label{proofsection}

\subsection{Proof of lemma \ref{teo:class2}}

We write $S=\begin{pmatrix}{} A & B \\ C & D \end{pmatrix}$ and express the condition of the vanishing off-diagonal terms 
\begin{equation}
A\varGamma_{1}C^{T}+B\varGamma_{2}D=0 \quad\text{for all}\quad \varGamma_{1},\varGamma_{2}\geq i\sigma.
\label{eq:bdt}
\end{equation}
If we fix $A,B,C,$ and $D$ all different from zero such that Eq.(\ref{eq:bdt}) is satisfied, we can always find $\varGamma_{1}$ 
and $\varGamma_{2}$ such that for these choice of submatrices $A\varGamma_{1}C^{T}+B\varGamma_{2}D^{T}\neq0$. Then each
summand in Eq.(\ref{eq:bdt}) must be zero. For instance if $A\varGamma_{1}C^{T}=-B\varGamma_{2}D^{T}$ take now 
$\varGamma_{1}\to2\varGamma_{1}$ then $A(2\varGamma_{1})C^{T}+B\varGamma_{2}D^{T}=A\varGamma_{1}C^{T}=0$. \\
W.l.o.g let us consider $A\varGamma_{1}C^{T}=0$, the other case will be analogous. We take the singular value decomposition 
of $A=UDV$ and $C=W\Sigma Z$ where $U,V,W,Z$ are unitaries. From here we choose $\varGamma_{1}=V^{-1}PZ^{-T}$ where 
$P$ is a positive definite matrix (recall that the sigma positive condition $\varGamma_{1}\geq i\sigma$ can be always 
obtained from rescaling a positive matrix). Then $A\varGamma_{1}C^{T}=UDP\Sigma W^{T}\neq0$ everytime 
we choose the proper $\varGamma_{1}$. Consequently $A=0$, $C=0$ or both. Likewise for $B$ and $D$. \\ 
The only matrices that fulfill the symplectic conditions are
\begin{equation}
\begin{pmatrix}{} A & 0 \\ 0 & D \end{pmatrix},\begin{pmatrix}{} 0 & B \\ C & 0 \end{pmatrix},
\end{equation}
provided $A,B,C$ and $D$ are symplectic.

\subsection{Proof of lemma \ref{teo:mainclass}}

First, it should be noted that the given assumptions immediately imply $S\begin{pmatrix}{} \varGamma & 0 \\ 0 & \varGamma \end{pmatrix}S^{T}=\begin{pmatrix}{} \cdot & 0 \\ 0 & * \end{pmatrix}$ 
for all $\varGamma>0$ and by continuity for all $\varGamma\geq0$ and consequently for all $\varGamma=\varGamma^{T}$. The latter is due
to the fact that every symmetric matrix can be decomposed in a semi-definite positive and negative part.\\
From Lemma \ref{l:lppm} we know which is the block structure of $S$. We consider the case where $A,B,C$ and $D$ are invertible, the other cases
will be contained here as we will see.
Using part (ii) of Lemma \ref{l:lppm} we have
\begin{equation}
A\otimes C=-B\otimes D.
\label{eq:77}
\end{equation}
We multiply Eq. (\ref{eq:77}) by $(A^{-1}\otimes\iden)$ from the left and take the trace in the first component, likewise we multiply
Eq. (\ref{eq:77}) by $(\iden\otimes D^{-1})$ and take the trace in the second component to obtain

\begin{equation}
C=\alpha D \quad \text{where}\quad \alpha:=\frac{-\Tr A^{-1}B}{2n}\in\mathbb{R},
\end{equation}
and 
\begin{equation}
B=\beta A \quad \text{where}\quad \beta:=\frac{-\Tr D^{-1}C}{2n}\in\mathbb{R}.
\end{equation}
From equation Eq. (\ref{eq:77}) we obtain that $\alpha=-\beta$\ and that
\begin{equation*}
S=\begin{pmatrix}{} A & \alpha A \\ -\alpha D & D \end{pmatrix}.
\end{equation*}
Moreover, the symplectic constraints on $S$ give us that $A,D\in\frac{1}{\sqrt{1+\alpha^{2}}}Sp(2n,\mathbb{R})$.
We write $A=\frac{1}{\sqrt{1+\alpha^{2}}}X$ and $D=\frac{1}{\sqrt{1+\alpha^{2}}}Y$ with $X,Y\in Sp(2n,\mathbb{R})$ to obtain 
Eq. (\ref{eq:main}). Finally for $\alpha\neq0$ we define $\gamma=\frac{1}{\alpha}$ and obtain the remaining equation. 
The case $\gamma=0$ is covered by $\alpha\to\pm\infty$ and it gives the swap operation.
Clearly $\alpha=0$ gives the local transformation.

%%%%%%%%%%%%%%%%%%%%%%%%%%%%%%%%%%%%%%%%%%%%%%%%%%%%%%%%%%%%%%%%%%%%%%%%%%%%%%%%%%%%%%%%%%%%%%%%%%%%%%%%%%%%%%%%%%%%%%%%%%%%%%%%%%%%%%%%%%%%%%%%%%%
%%%%%%%%%%%%%%%%%%%%%%%%%%%%%%%%%%%%%%%%%%%%%%%%%%%%%%%%%%%%%%%%%%%%%%%%%%%%%%%%%%%%%%%%%%%%%%%%%%%%%%%%%%%%%%%%%%%%%%%%%%%%%%%%%%%%%%%%%%%%%%%%%%%
%%%%%%%%%%%%%%%%%%%%%%%%%%%%%%%%%%%%%%%%%%%%%%%%%%%Proof oF STABILITY%%%%%%%%%%%%%%%%%%%%%%%%%%%%%%%%%%%%%%%%%%%%%%%%%%%%%%%%%%%%%%%%%%%%%%%%%%%%%%%%%%%
%%%%%%%%%%%%%%%%%%%%%%%%%%%%%%%%%%%%%%%%%%%%%%%%%%%%%%%%%%%%%%%%%%%%%%%%%%%%%%%%%%%%%%%%%%%%%%%%%%%%%%%%%%%%%%%%%%%%%%%%%%%%%%%%%%%%%%%%%%%%%%%%%%%
%%%%%%%%%%%%%%%%%%%%%%%%%%%%%%%%%%%%%%%%%%%%%%%%%%%%%%%%%%%%%%%%%%%%%%%%%%%%%%%%%%%%%%%%%%%%%%%%%%%%%%%%%%%%%%%%%%%%%%%%%%%%%%%%%%%%%%%%%%%%%%%%%%%

\subsection{Proof of the stability of DS theorem \ref{teo:pids}} \label{proofofteo:pids}

The proof involves a series of steps. We first use Parseval's theorem to express the distance of the quantum states in terms of the $L_{2}-$distance of their respective characteristic functions. Next, we show that there is a ball $\mathfrak{B}_{r}$ around the origin of phase space where the characteristic functions do not vanish. The radius of this ball scales inversely proportional to the largest variance of the input states and proportional to $\log(1/\varepsilon)$. Thus the smaller the error parameter $\varepsilon$ the bigger this region is. We then proceed to bound the distance separately on $\mathfrak{B}_{r}$ and its complement $\mathfrak{B}^{c}_{r}$. For the latter region, we exploit the relation between the tails of a distribution and the finiteness of it moments. Inside $\mathfrak{B}_{r}$, the problem is equivalent to the stability of the Gaussian functional Eq.~\eqref{eq:fegs} with restricted convex domain. For that matter, the stability of the Gaussian functional equation is reduced to the stability of the fundamental functional equation, namely the Cauchy functional equation.   \\

We may assumed without lost of generality that the input states (and therefore the output states) are centered. This is justified by the fact that the trace and HS norms are unitarily invariant and the operation of ``Gaussification" (the operation on bosonic quantum states which produces Gaussian states with the same first and second moments of the input state) commutes with the displacement operation $\rho\mapsto W_{d}\rho W^{*}_{d}$. \\

 We denote by $\chi_{ab},\chi_{a}$ and $\chi_{b}$ the characteristic functions of $\rho_{ab},\rho_{a}$ and $\rho_{b}$, respectively. From Lemma~\eqref{lemma:closetraces}, we have 
 $\norm{g}_{1}\leq3\varepsilon$. Let $\eta_{1},\eta_{2}\in\reals^{2n}$ and $R_{1},R_{2}$ be vectors of canonical operators as in Eq. \ref{eq:notationR1R2}.
We write
\begin{equation}
 G(\eta_{1},\eta_{2}):=\Tr[e^{i\eta_{1}\cdot \sigma R_{1}}\otimes e^{i\eta_{2}\cdot \sigma R_{2}}g]=\chi_{ab}(\eta_{1},\eta_{2})-\chi_{a}(\eta_{1})\chi_{b}(\eta_{2}).
\end{equation}
Now, using the covariant property of Gaussian unitary operations Eq.~\eqref{eq:cov}

\begin{equation*}
 \chi_{\rho_{ab}}(\xi)=\chi_{\rho_{1}\otimes\rho_{2}}(S^{T}_{\theta}\xi),
\end{equation*}

and the fact that $G(\eta_{1},0)=G(0,\eta_{2})=0$, we write for all $\eta_{1},\eta_{2}\in\reals^{2n}$ the dynamical process in terms 
of characteristic functions as 

\begin{equation}
  \chi_{1}(\cos\theta\eta_{1}+\sin\theta\eta_{2})\chi_{2}(\cos\theta\eta_{2}-\sin\theta\eta_{1})=\chi_{1}(\cos\theta\eta_{1})\chi_{1}(\sin\theta\eta_{2})\chi_{2}(\cos\theta\eta_{2})\chi_{2}(-\sin\theta\eta_{1})
+G(\eta_{1},\eta_{2}).
\label{eq:pfegs}
 \end{equation}

This last equation resembles the ideal functional equation Eq.~\eqref{eq:fegs} plus a new remainder term $G(\eta_{1},\eta_{2})$. From  H\"older's inequality and the
definition of $G$ we note that $\norm{G}\leq3\varepsilon$.\\

We state the following lemma with proof in section \ref{technicallemmassubsection}.

\begin{lemma}\label{lemma:region}
 Let $\chi_{1}$ and $\chi_{2}$ be two characteristic functions with respective density operators $\rho_1$, $\rho_2$ and
 covariance matrices $\varGamma_{1},\varGamma_{2}$. Define for any $\theta\notin\{\mathbb{Z}\frac{\pi}{2}\}$
 \begin{equation*}
  G(\eta_{1},\eta_{2}):=\chi_{1}(\cos\theta\eta_{1})\chi_{1}(\sin\theta\eta_{2})\chi_{2}(\cos\theta\eta_{2})\chi_{2}(-\sin\theta\eta_{1})-\chi_{1}(\cos\theta\eta_{1}+\sin\theta\eta_{2})\chi_{2}(\cos\theta\eta_{2}-\sin\theta\eta_{1}).
\end{equation*}
Assume $|G(\eta_1,\eta_2)|\leq 3\epsilon$ for all $\eta_{1},\eta_{2}\in\reals^{2n}$, let $\lambda:=\frac{1}{2}\max\{\norm{\varGamma_{1}},\norm{\varGamma_{2}}\}$ be the largest variance of the states $\rho_{1}$

% another way to express the CM is
% $\lambda:=\max\big\{ \Tr[R^{2}_{\xi}\rho_{j}]-\left(\Tr[R_{\xi}\rho_{j}]\right)^{2}\;\big|\;\|\xi\|=1,\ j=1,2\big\}$ 

 and $\rho_{2}$ and define
 \begin{equation}
  r:=\sqrt{\frac{1}{\lambda}\log_{2}\frac{1}{\varepsilon^{1/12}}}.
  \label{def:r}
 \end{equation}

 Then for $\eta\in\mathfrak{B}_{r}:=\left\{\xi\ | \norm{\xi}_{2}\leq r\right\}$
 \begin{equation}
  |\chi_{i}(\eta)|\geq 12\varepsilon^{1/12} \qquad\text{i=1,2.}
 \end{equation}
%\label{eq:region}
\end{lemma}
The choice of exponent ${1/12}$ for $\varepsilon$ in Eq.~\eqref{def:r} will become clear in the estimates done in section \ref{sec:appendix}.\\

We divide phase space in two separating regions. One where the characteristic functions do not vanish, namely inside the ball $\mathfrak{B}_{r}:=\left\{\xi\ | \norm{\xi}^{2}_{2}\leq r^{2}\right\}$, 
and its complement which we denote by $\mathfrak{B}^{c}_{r}$. So that 
with the help of Parseval's theorem, we express the distance between the two states as 
\begin{equation}
\begin{split}
 (2\pi)^{n}\norm{\rho_{1}-\rho'_{1}}^{2}_{2}&=\norm{\chi_{1}-\Phi}^{2}_{2}\\
 &=\int\limits_{\xi\in\mathfrak{B}_{r/2}}\left|\chi_{1}(\xi)-\Phi(\xi)\right|^{2} d\xi +\int\limits_{\xi\in\mathfrak{B}^{c}_{r/2}}\left|\chi_{1}(\xi)-\Phi(\xi)\right|^{2} d\xi.
\label{eq:Parseval2}
 \end{split}
\end{equation}

We compute the bound for $\mathfrak{B}_{r/2}$ and $\mathfrak{B}^{c}_{r/2}$ separately.\\

{\renewcommand\addcontentsline[3]{} \subsubsection*{Bound on the region where the characteristic function might vanish}} \label{subsec:2}
%\subsubsection*{Bound on the region where the characteristic function might vanish}

We use $|z_{1}-z_{2}|^{2}\leq (|z_{1}|^{2}+|z_{2}|^{2})/2$ for $z_{1},z_{2}\in\mathbb{C}$ to express the bound as 

\begin{equation} \label{eq:tailseq}
  \int\limits_{\xi\in\mathfrak{B}^{c}_{r/2}} |\chi_{1}(\xi)-\Phi(\xi)|^{2}d\xi \leq \frac{1}{2}\int\limits_{\xi\in\mathfrak{B}^{c}_{r/2}}|\chi_{1}(\xi)|^{2}d\xi + \frac{1}{2}\int\limits_{\xi\in\mathfrak{B}^{c}_{r/2}}|\Phi(\xi)|^{2}d\xi.
\end{equation}

For the first term of the RHS of Eq. (\ref{eq:tailseq}) we use that $1<2\norm{\xi}_{2}/r$ for $\xi\in\mathfrak{B}^{c}_{r/2}$ so that 

\begin{equation}\label{eq:bddtail1}
 \int\limits_{\xi\in\mathfrak{B}^{c}_{r/2}}|\chi_{1}(\xi)|^{2}d\xi \leq \frac{4}{r^{2}}\int\limits_{\xi\in\reals^{2n}}\norm{\xi}^{2}_{2}|\chi_{1}(\xi)|^{2}d\xi.
\end{equation}

Let us denote by $\mathcal{W}(\eta)$ the Wigner function of $|\chi_{1}(\xi)|^{2}$ (the product of two characteristic functions is a characteristic function), 
\begin{equation}\label{eq:wigrepsymcf}
\mathcal{W}(\eta)=\frac{1}{(2\pi)^{2n}}\int\limits_{\xi\in\reals^{2n}}e^{i\eta\cdot\sigma\xi}|\chi_{1}(\xi)|^{2}d\xi.
\end{equation}

It can be easily verified by direct computation, that the characteristic function $|\chi_{1}(\xi)|^{2}$ is centered and has CM $2\varGamma_{1}$. Moreover, we have

\begin{equation}\label{eq:bddtail2}
\int\limits_{\xi\in\reals^{2n}}\norm{\xi}^{2}_{2}|\chi_{1}(\xi)|^{2}d\xi = -(2\pi)^{n}\sum_{k=1}^{2n}\frac{\partial^{2}\mathcal{W}(0)}{\partial\eta^{2}_{k}},
\end{equation}
where $\eta_{k}$ are the components of the vector $\eta$ in an arbitrary but fixed basis. 

We use now the representation of the Wigner function in terms of the expectation values of the parity operator\cite{Gro76} $\mathcal{P}$

\begin{equation}\label{eq:parityrep}
 \mathcal{W}(\eta)=\frac{1}{\pi^{n}}\Tr[\rho W_{\eta}\mathcal{P}W_{-\eta}],
\end{equation}
where $\rho$ is the density operator corresponding to the characteristic function $|\chi_{1}(\xi)|^{2}$. The operator $\rho$ is clearly Schwartz as its characteristic function is a Schwartz function (corollary \ref{cor:schwtrace}). The parity operator $\mathcal{P}$ is the n-fold tensor product of the parity operators for a single degree of freedom and is the unitary operator that satisfies
\begin{align*}
&\mathcal{P}W_{\xi}\mathcal{P}^{*}=W_{-\xi}, \\
&\mathcal{P}R_{k}\mathcal{P}^{*}=-R_{k}, \\
 &\mathcal{P}=\mathcal{P}^{*}=\mathcal{P}^{-1}.
\end{align*}

Using Eq.~\eqref{eq:parityrep} and Eq.~\eqref{eq:gradient2} we compute

\begin{equation}\label{eq:bddtail3}
 \frac{\partial^{2}\mathcal{W}(0)}{\partial\eta_{k}\partial\eta_{l}}=-\frac{2}{\pi^{n}}\Tr[\mathcal{P}\rho\{R_{\eta_{l}},R_{\eta_{k}}\}].
\end{equation}

Thus from Eq.~\eqref{eq:bddtail1},~\eqref{eq:bddtail2} and Eq.~\eqref{eq:bddtail3} we find that

\begin{equation}
\int\limits_{\xi\in\mathfrak{B}^{c}_{r/2}}|\chi_{1}(\xi)|^{2}d\xi \leq  \frac{2^{n+4}}{r^{2}}\sum_{k=1}^{2n}\Tr [\mathcal{P}\rho R^{2}_{k}].
\end{equation}
Hence, we just need to bound the terms $\Tr [\mathcal{P}\rho R^{2}_{k}]$. In order to do this, we notice that $\mathcal{P}\rho=\rho\mathcal{P}$. Indeed,

\begin{align*}
\mathcal{P}\rho=\mathcal{P}\frac{1}{(2\pi)^{n}}\int |\chi(\xi)|^{2} W_{-\xi}d\xi &=\frac{1}{(2\pi)^{n}}\int |\chi(\xi)|^{2} W_{\xi}d\xi\;\mathcal{P}, \\
&=\frac{1}{(2\pi)^{n}}\int |\chi(\xi)|^{2} W_{-\xi}d\xi\;\mathcal{P}=\rho\mathcal{P}.
\end{align*}

Moreover, from the spectral decomposition of $\rho$ we have $\mathcal{P}\sqrt{\rho}=\sqrt{\rho}\:\mathcal{P}$. Accordingly,

\begin{align*}
\Tr [\mathcal{P}\rho R^{2}_{k}]&= \Tr [\mathcal{P}\sqrt{\rho} R^{2}_{k} \sqrt{\rho}],\\
&\leq\norm{\sqrt{\rho} R^{2}_{k} \sqrt{\rho}}_{1}, \\
&= \Tr[\rho R_{k}^{2}].
\end{align*}
Here we have used the Cauchy-Schwarz inequality and the cyclicity of the trace that comes from the properties of the Schwartz operator $\rho$; see Theorem \ref{theo:schwprop} $(i),(ii),(iv)$.\\

 In summary, we have the following bound for the tails of our characteristic function 
 \begin{align*}
   \int\limits_{\xi\in\mathfrak{B}^{c}_{r/2}}|\chi_{1}(\xi)|^{2}d\xi &\leq \frac{2^{n+4}}{r^{2}}\Tr \varGamma_{1}, \\
  & =\left(2^{n+4}12\lambda\Tr \varGamma_{1}\right)\frac{1}{\log\frac{1}{\varepsilon}}.
 \end{align*}

Since $\Phi(\xi)$ has the same CM as $\chi(\xi)$ we can use the same bound to obtain 
 
 \begin{align}
 \frac{1}{(2\pi)^{n}}\int\limits_{\xi\in\mathfrak{B}^{c}_{r/2}} |\chi_{1}(\xi)-\Phi(\xi)|^{2}d\xi &\leq \left(\frac{192\lambda\Tr \varGamma_{1}}{\pi^{n}}\right)\frac{1}{\log\frac{1}{\varepsilon}}, \nonumber \\
 &\leq \frac{c^{2}_{2}}{\log\frac{1}{\varepsilon}}, \label{eq:bddoutside}
 \end{align}
where
\begin{equation}
 c_{2}:=8\sqrt{\frac{3\lambda\Tr \varGamma_{ab}}{\pi^{n}}},
\end{equation}

and $\varGamma_{ab}$ is the CM of the output state $\rho_{ab}$. Note that since the BS is a passive transformation, the trace of the input CM, $\varGamma_{1}\oplus\varGamma_{2}$, is the same as the trace of the output CM $\varGamma_{ab}$. Thus is clear that $\Tr\varGamma_{1}\leq\Tr\varGamma_{ab}$. \\

{\renewcommand\addcontentsline[3]{}  \subsubsection*{Bound inside the region where $\chi$ does not vanish:}}

We proceed to compute the bound for the first term of the RHS of Eq. (\ref{eq:Parseval2}) following the ideas of Ref.~\onlinecite{fisk69,KY85,Ga74}. For that matter, let 
$\eta_{1},\eta_{2}\in\reals^{2n}$ with $\norm{\eta_{j}}_{2}<r/2$ so that $\cos\theta\eta_{1}+\sin\theta\eta_{2},\cos\theta\eta_{2}-\sin\theta\eta_{1}\in\mathfrak{B}_{r}$. 
Let us take the logarithm (principal branch) on both sides of Eq. (\ref{eq:pfegs}). We write
\begin{equation}
\Psi_{1}(\cos\theta\eta_{1}+\sin\theta\eta_{2})+\Psi_{2}(\cos\theta\eta_{2}-\sin\theta\eta_{1})=\Psi_{1}(\cos\theta\eta_{1})+\Psi_{1}(\sin\theta\eta_{2})+\Psi_{2}(\cos\theta\eta_{2})+\Psi_{2}(-\sin\theta\eta_{1})+Q(\eta_{1},\eta_{2}),
\label{eq:lfe}
\end{equation}
where $\Psi_{j}(\eta):=-\log\chi_{j}(\eta)$ for $j=1,2$ and 
\begin{equation}
Q(\eta_{1},\eta_{2}):=-\log\left(1+\frac{G(\eta_{1},\eta_{2})}{\chi_{1}(\cos\theta\eta_{1})\chi_{1}(\sin\theta\eta_{2})\chi_{2}(\cos\theta\eta_{2})\chi_{2}(-\sin\theta\eta_{1})}\right).
\label{eq:q}
\end{equation}

Since $\rho_{j},j=1,2$ is Schwartz, we can define continuous vector-valued functions $\phi_{j}(\xi):\mathfrak{B}_{r/2}\to\mathbb{C}^{2n}$ by 
\begin{equation*}
\phi_{j}(\xi):=\nabla\Psi_{j}(\xi).
\end{equation*}
 Note that $\phi_{j}(\xi),j=1,2,$ are in fact conservative vector fields and that $\phi_{j}(0)=0$ as $\rho_{j}$ is centered.  The gradient of $\phi_{j}$ is the Hessian of $\chi_{j}$ (see Lemma \ref{l:qcfmoments2}) and $2\nabla\phi_{j}(0)=\sigma\varGamma_{j}\sigma^{T}$ for $j=1,2$.  \\

{\renewcommand\addcontentsline[3]{}  \subsubsection*{Inhomogeneous Cauchy Functional Equation}}

Next, we want to obtain a functional equation only depending on $\chi_{1}$ or $\chi_{2}$. In order to 
do so, we differentiate Eq.~\eqref{eq:lfe} in the direction of $\eta_{1}$ to find
\begin{equation}
 \cos\theta\phi_{1}(\cos\theta\eta_{1}+\sin\theta\eta_{2})-\sin\theta\phi_{2}(\cos\theta\eta_{2}-\sin\theta\eta_{1})=\cos\theta\phi_{1}(\cos\theta\eta_{1})-
 \sin\theta\phi_{2}(-\sin\theta\eta_{1})+Q_{1}(\eta_{1},\eta_{2}),
\label{eq:ec9}
 \end{equation}
where $Q_{1}(\eta_{1},\eta_{2}):=\frac{dQ(\eta_{1}+t\eta_{1},\eta_{2})}{dt}\Bigr|_{t=0}$. We evaluate in Eq.~\eqref{eq:ec9} $\eta_{1}=0$ to get
\begin{equation}
 \cos\theta\phi_{1}(\sin\theta\eta_{2})-\sin\theta\phi_{2}(\cos\theta\eta_{2})=Q_{1}(0,\eta_{2}),
\label{eq:ec10}
\end{equation}
In a similar fashion we differentiate Eq.~\eqref{eq:lfe} in the direction of $\eta_{2}$ and set to zero to obtain

\begin{equation}
 \sin\theta\phi_{1}(\cos\theta\eta_{1}+\sin\theta\eta_{2})+\cos\theta\phi_{2}(\cos\theta\eta_{2}-\sin\theta\eta_{1})=\sin\theta\phi_{1}(\sin\theta\eta_{2})+
 \cos\theta\phi_{2}(\cos\theta\eta_{2})+Q_{2}(\eta_{1},\eta_{2}),
\label{eq:ec11}
 \end{equation}

\begin{equation}
 \sin\theta\phi_{1}(\cos\theta\eta_{1})+\cos\theta\phi_{2}(-\sin\theta\eta_{1})=Q_{2}(\eta_{1},0),
\label{eq:ec12}
\end{equation}
where 
%\begin{equation*}
%Q_{2}(\eta_{1},\eta_{2}):=\frac{dQ(\eta_{1},\eta_{2}+t\eta_{2})}{dt}\Bigr|_{t=0}
%\end{equation*}
$Q_{2}(\eta_{1},\eta_{2}):=\frac{dQ(\eta_{1},\eta_{2}+t\eta_{2})}{dt}\Bigr|_{t=0}$. \\

Now we will be able to decouple $\phi_{1}$ and $\phi_{2}$. First, we substract Eq.~\eqref{eq:ec10} from Eq.~\eqref{eq:ec9} and Eq.~\eqref{eq:ec12} from Eq.~\eqref{eq:ec11} to obtain

\begin{equation}
\begin{split}
 [\phi_{1}(\cos\theta\eta_{1}+\sin\theta\eta_{2})-\phi_{1}(\cos\theta\eta_{1})-\phi_{1}(\sin\theta\eta_{2})]&=\tan\theta[\phi_{2}(\cos\theta\eta_{2}-\sin\theta\eta_{1})-\phi_{2}(\cos\theta\eta_{2})-\phi_{2}(-\sin\theta\eta_{1})]\\
&+\frac{Q_{1}(\eta_{1},\eta_{2})-Q_{1}(0,\eta_{2})}{\cos\theta},
\end{split}
\label{eq:ec13}
\end{equation}
 
\begin{equation}
 \begin{split}
 [\phi_{2}(\cos\theta\eta_{2}-\sin\theta\eta_{1})-\phi_{2}(\cos\theta\eta_{2})-\phi_{2}(-\sin\theta\eta_{1})]&=-\tan\theta[\phi_{1}(\cos\theta\eta_{1}+\sin\theta\eta_{2})-\phi_{1}(\cos\theta\eta_{1})-\phi_{1}(\sin\theta\eta_{2})]\\
&+\frac{Q_{2}(\eta_{1},\eta_{2})-Q_{2}(\eta_{1},0)}{\cos\theta}.
\end{split} 
\label{eq:ec14}
 \end{equation} 

Thus from Eq.~\eqref{eq:ec13} and \eqref{eq:ec14} we find the following inhomogeneous Cauchy equations

\begin{equation}
\begin{split}
 [\phi_{1}(\cos\theta\eta_{1}+\sin\theta\eta_{2})-\phi_{1}(\cos\theta\eta_{1})-\phi_{1}(\sin\theta\eta_{2})]&=\sin\theta[Q_{2}(\eta_{1},\eta_{2})-Q_{2}(\eta_{1},0)]\\
&+\cos\theta[Q_{1}(\eta_{1},\eta_{2})-Q_{1}(0,\eta_{2})],
\end{split}
\label{eq:ice1}
\end{equation}

\begin{equation}
\begin{split}
 [\phi_{2}(\cos\theta\eta_{2}-\sin\theta\eta_{1})-\phi_{2}(\cos\theta\eta_{2})-\phi_{2}(-\sin\theta\eta_{1})]&=\sin\theta[Q_{1}(\eta_{1},\eta_{2})-Q_{1}(0,\eta_{2})]\\
&+\cos\theta[Q_{2}(\eta_{1},\eta_{2})-Q_{2}(\eta_{1},0)].
\end{split}
\label{eq:ice2}
\end{equation}

{\renewcommand\addcontentsline[3]{} \subsubsection*{Bound on $\mathfrak{B}_{r/2}$}}

Now that $\phi_{1}$ and $\phi_{2}$ are decoupled, we continue only with $\phi_{1}$ as with $\phi_{2}$ is analogous and the same upper bound is obtained. We recall that the derivative of a vector with respect to a vector can be represented as a matrix. Thus when we differentiate Eq.~\eqref{eq:ice1} in the direction of $\eta_{2}$ and evaluate at $\eta_{2}=0$, we obtain the following matrix-valued equation

\begin{equation}
\nabla\phi_{1}(\cos\theta\eta_{1})-\frac{\sigma\varGamma_{1}\sigma^{T}}{2}+\frac{Q_{12}(0,0)}{\tan\theta}= Q_{22}(\eta_{1})+\frac{Q_{12}(\eta_{1})}{\tan\theta}.
\end{equation} 
Here $Q_{12}(\eta_{1},0):\reals^{2n}\to\mathbb{C}^{2n\times2n}$ and $Q_{22}(\eta_{1},0):\reals^{2n}\to\mathbb{C}^{2n\times2n}$ are defined as 
\begin{equation}\label{eq:}
 Q_{12}(\eta_{1},0):=\frac{\partial^{2}Q(\eta_{1}+t\eta_{1},s\eta_{2})}{\partial s \partial t}\Bigr|_{s=t=0} \qquad\text{and}\qquad Q_{22}(\eta_{1}):=\frac{\partial^{2}Q(\eta_{1},s\eta_{2}+t\eta_{2})}{\partial s \partial t}\Bigr|_{s=t=0}.
\end{equation}

 Accordingly, we integrate the previous equation twice from zero to $\eta$, $\norm{\eta}_{2}\leq r/2$. We obtain for $\xi\in\mathfrak{B}_{r/2}$

\begin{equation}
 \chi_{1}(\xi)=\exp[-\xi\cdot\left(\frac{\sigma\varGamma_{1}\sigma^{T}}{4}-\frac{V}{2}\right)\xi-F\left(\frac{\xi}{\cos\theta}\right)] ,
 \label{eq:sol1}
\end{equation}
with 
\begin{equation*}
 V:=-\frac{Q_{12}(0,0)}{\tan\theta}=\frac{\sigma(\Tr gR_{1}R_{2}^{T})\sigma^{T}}{\tan\theta},
\end{equation*}
\begin{equation*}
F(\xi):= \cos^{2}\theta\int_{\mathcal{C}(\xi)}\left( \int_{\mathcal{C}(\eta)} \left(Q_{22}(\eta_{1})+\frac{Q_{12}(\eta_{1})}{\tan\theta}\right)\cdot d\eta_{1}\right)\cdot d\eta,
\end{equation*}
where $\mathcal{C}(\xi),\mathcal{C}(\eta)$ are curves in phase space connecting the origin with the vectors $\xi$ and $\eta$. Moreover, these last terms can be upper bounded by (see Appendix on section \ref{sec:appendix})
\begin{align}
\norm{V}_{2}&\leq \left(\frac{\sqrt{24 n^{2}\kappa}}{|\tan\theta|}\right)\sqrt{\varepsilon},\label{eq:bddV}\\
\left|F\left(\frac{\xi}{\cos\theta}\right)\right|^{2} &\leq \left(\frac{n^{2}\kappa \norm{\xi}_{2}^{4}}{2\tan^{2}\theta}\right)\varepsilon^{2/3}, \label{eq:bdddF}
\end{align}

where the largest absolute fourth moment of $\rho_{ab}$ is defined as 
\begin{equation}
\kappa:= \max\left\{ \norm{\rho_{ab}R^{2}_{\xi}R^{2}_{\eta}}_{1}\;\big|\;\norm{\xi}_{2}=\norm{\eta}_{2}=1\right\}.
\end{equation}

At this point we can show that the CM of $\rho_{1}$ and $\rho_{2}$ are $\varepsilon-$close. From differentiating Eq. \eqref{eq:ec10} with respect to $\eta_{2}$ and evaluating at zero, 
we get the relation between the CMs of $\rho_{1}$ and $\rho_{2}$,

\begin{equation*}
\varGamma_{1}-\varGamma_{2}=\frac{2}{\cos^{2}\theta}V. 
\end{equation*}
Hence, from Eq.~\eqref{eq:bddV}

\begin{equation}
 \norm{\varGamma_{1}-\varGamma_{2}}_{2}\leq \left(\frac{\sqrt{384n^{2}\kappa}}{|\sin2\theta|}\right) \ \sqrt{\varepsilon}.
\end{equation}

Now we can proceed to show that for $\xi\in\mathfrak{B}_{r/2}$, the characteristic function of the state $\rho_{1}$ is $\varepsilon-$close to the Gaussian characteristic 
function $\Phi=\exp[-\xi\cdot\varGamma_{1}\xi/4]$. \\

We use Eq.~\eqref{eq:sol1}, $|e^{x}-1|\leq |x|\max\left\{1,e^{\operatorname{Re}[x]}\right\}$  and $|\chi_{1}(\xi)|^{2}=|\Phi(\xi)|^{2}e^{\xi\cdot V\xi-2\operatorname{Re}[F(\xi/\cos\theta)]}$ to write the bound as 

\begin{align}
 \int\limits_{\xi\in\mathfrak{B}_{r/2}}\left|\chi_{1}(\xi)-\Phi(\xi)\right|^{2}d\xi &= \int\limits_{\xi\in\mathfrak{B}_{r/2}}\left|\Phi(\xi)|^{2}|\exp[\xi\cdot V\xi/2-F(\xi/\cos\theta)]-1\right|^{2}d\xi \nonumber\\
 &\leq\int\limits_{\xi\in\mathfrak{B}_{r/2}}|\chi_{1}(\xi)|^{2}\left|\frac{\xi\cdot V\xi}{2}-F(\xi/\cos\theta)\right|^{2}d\xi, \nonumber\\
  &\leq \frac{1}{2}\int\limits_{\xi\in\mathfrak{B}_{r/2}} |\chi_{1}(\xi)|^{2} \left(\frac{|\xi\cdot V\xi|^{2}}{4} + |F(\xi/\cos\theta)|^{2}\right) d\xi  \nonumber \\
 &\leq \left(\frac{4n^{2}\kappa\varepsilon^{2/3}}{\tan^{2}\theta}\right)\int\limits_{\xi\in\mathfrak{B}_{r/2}} |\chi_{1}(\xi)|^{2} \norm{\xi}_{2}^{4} d\xi,  \label{eq:ibdd}
\end{align}
Here in the second inequality we have used again $|z_{1}-z_{2}|^{2}\leq (|z_{1}|^{2}+|z_{2}|^{2})/2$ for $z_{1},z_{2}\in\mathbb{C}$ and in the last inequality Eq.~\eqref{eq:bddV} and Eq.~\eqref{eq:bdddF}. We show now that 

\begin{equation}\label{eq:bd4moments}
\frac{1}{(2\pi)^{n}}\int\limits_{\xi\in\mathfrak{B}_{r/2}} |\chi_{1}(\xi)|^{2} \norm{\xi}_{2}^{4} d\xi \leq \frac{512n^{2}\kappa}{\pi^{n}}\left(\frac{1+3\sin2\theta}{\cos^{4}\theta}\right).
\end{equation}

We used again the Wigner representation, Eq.\eqref{eq:wigrepsymcf}, of the characteristic function $|\chi_{1}(\xi)|^{2}$ in order to compute the bound of Eq. \eqref{eq:bd4moments}
\begin{align*}
\frac{1}{(2\pi)^{n}} \int\limits_{\xi\in\mathfrak{B}_{r/2}} |\chi_{1}(\xi)|^{2} \norm{\xi}_{2}^{4}  d\xi &\leq \frac{1}{(2\pi)^{n}} \int\limits_{\xi\in\reals^{2n}} |\chi_{1}(\xi)|^{2} \norm{\xi}_{2}^{4}  d\xi, \\
&=\frac{1}{(2\pi)^{n}} \sum_{k,l=1}^{2n}\ \int\limits_{\xi\in\reals^{2n}} |\chi_{1}(\xi)|^{2} \xi^{2}_{k}\xi^{2}_{l} d\xi , \\
&= \frac{\partial^{4}\mathcal{W}(0)}{\partial\eta^{2}_{k}\partial\eta^{2}_{l}}.
\end{align*}

Using Lemma \ref{l:qcfmoments2} with the help of the CCR gives

\begin{equation*}
 \sum_{k,l=1}^{2n} \frac{\partial^{4}\mathcal{W}(0)}{\partial\eta^{2}_{k}\partial\eta^{2}_{l}}= \frac{8}{\pi^{n}} \sum_{k,l=1}^{2n} \Tr[\rho\mathcal{P}\{R^{2}_{k},R^{2}_{l}\}].
\end{equation*}

Here $\rho$ is again the density operator corresponding to the characteristic function $|\chi(\xi)|^{2}$. We use again that $\mathcal{P}\rho=\rho\mathcal{P}$ (see subsection \ref{subsec:2}) and H\"older's inequality to bound

\begin{equation*}
\frac{8}{\pi^{n}} \sum_{k,l=1}^{2n} \Tr[\rho\mathcal{P}\{R^{2}_{k},R^{2}_{l}\}] \leq \frac{8}{\pi^{n}} \sum_{k,l=1}^{2n} \left(\norm{\sqrt{\rho}R^{2}_{k}R^{2}_{l}\sqrt{\rho}}_{1}+\norm{\sqrt{\rho}R^{2}_{l}R^{2}_{k}\sqrt{\rho}}_{1}\right).
\end{equation*}

Now using Cauchy-Schwartz inequality and the ciclycity properties for Schwartz operators we find 

\begin{align}
\frac{8}{\pi^{n}} \sum_{k,l=1}^{2n} \Tr[\rho\mathcal{P}\{R^{2}_{k},R^{2}_{l}\}] &\leq \frac{16}{\pi^{n}}\left( \sum^{2n}_{k=1}(\Tr\rho R^{4}_{k})^{1/2} \right)^{2}, \nonumber \\
&\leq  \frac{32n}{\pi^{n}} \sum^{2n}_{k=1}\Tr\rho R^{4}_{k}.  \label{eq:rel4m}
\end{align}

Since we want to specify all the constants in terms of the moments of the output state $\rho_{ab}$ we need to do the following computations. Using the explicit form of the BS transformation Eq.~\eqref{eq:bstheta} and the derivatives of $\chi_{1}(\xi)$, we obtain after a tedious, but straightforward calculation
\begin{equation*}
\cos^{4}\theta \sum^{2n}_{k=1}\Tr\rho_{1}R^{4}_{1k} +\sin^{4}\theta \sum^{2n}_{k=1}\Tr\rho_{2}R^{4}_{2k}-6\sin^{2}\theta\cos^{2}\theta\sum^{2n}_{k=1}\Tr\rho_{1}R^{2}_{1k}\Tr\rho_{2}R^{2}_{2k}= \sum^{2n}_{k=1}\Tr\rho_{ab}R^{4}_{k},
\end{equation*}
which together with the positivity of the fourth moments implies
\begin{equation}\label{eq:in_out_moments}
\sum^{2n}_{k=1}\Tr\rho_{1}R^{4}_{k}\leq 6\tan^{2}\theta\sum^{2n}_{k=1}\Tr\rho_{1}R^{2}_{1k}\Tr\rho_{2}R^{2}_{2k}+ \frac{1}{\cos^{4}\theta }\sum^{2n}_{k=1}\Tr\rho_{ab}R^{4}_{k}.
\end{equation}
Moreover, the relation between the fourth moments of the symmetrized state $\rho$ and $\rho_{1}$ is given by

\begin{equation}
\Tr\rho R^{4}_{k} =2\Tr\rho_{1}R^{4}_{k}+6(\Tr\rho_{1}R^{2}_{k})^{2}, \qquad k=1,\ldots,2n. \label{eq:bd4symm}
\end{equation}
Combining Eqs. \eqref{eq:in_out_moments} and \eqref{eq:bd4symm} and using  $(\Tr\rho_{1}R^{2}_{k})^{2}\leq \Tr\rho_{1}R^{4}_{k}\leq \kappa$ we obtain
\begin{equation*}
\sum_{k=1}^{2n}\Tr\rho R^{4}_{k} \leq 16n\kappa\left(\frac{1+3\sin2\theta}{\cos^{4}\theta}\right),
\end{equation*}
and Eq. \eqref{eq:rel4m} gives the claimed bound in Eq. \eqref{eq:bd4moments}. Hence, inserting Eq.~\eqref{eq:bd4moments} in Eq.~\eqref{eq:ibdd} we obtain the bound for the non-vanishing region $\mathfrak{B}_{r/2}$ 

\begin{equation}\label{eq:bddinside}
\frac{1}{(2\pi)^{n}}\int\limits_{\xi\in\mathfrak{B}_{r/2}}\left|\chi_{1}(\xi)-\Phi(\xi)\right|^{2}d\xi\leq c^{2}_{1}\varepsilon^{2/3},
\end{equation}
where 
\begin{equation}
c_{1}:=32\sqrt{\frac{2}{\pi^{n}}\left(\frac{1+3\sin2\theta}{\sin2\theta}\right)}n^{2}\kappa.
\end{equation}

\begin{figure}[h]\label{fig:theta}
\includegraphics[scale=1]{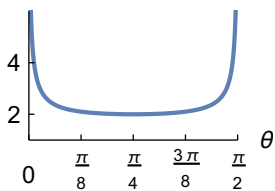}
\caption{Dependence of the stability constant $c_{1}$ in terms of the transmission coefficient $\theta$. The $y-$axis is the function $\sqrt{\frac{1+3\sin2\theta}{\sin2\theta}}$. The constant $c_{1}(\theta,n,\kappa)$ for the one mode 50-50 BS case is $c_{1}\approx 46.2 \kappa.$}
\end{figure}

The result of the theorem follows from Eq.~\eqref{eq:bddoutside}, \eqref{eq:bddinside} and the concavity of the square root.\\

%%%%%%
%%%%%%%%%%%%%%%%%%%%%%%%%%%%%%%%%%%%%%%%%%%%%%%%%%%%%%%%%%%%%%%%%%%%%%%%%%%%%%%%%%%%%%%%%%%%%%%%%%%%%%%%%%%%%%%%%%%%%%%%%%%%%%%%%%%%%%%%%%%%%%%%%%%
%%%%%%%%%%%%%%%%%%%%%%%%%%%%%%%%%%%%%%%%%%%%%%%%%%%%%%%%%%%%%%%%%%%%%%%%%%%%%%%%%%%%%%%%%%%%%%%%%%%%%%%%%%%%%%%%%%%%%%%%%%%%%%%%%%%%%%%%%%%%%%%%%%%
%%%%%%%%%%%%%%%%

\section{Discussion}

The DS theorem can be understood as the statement that Gaussian bosonic states with same covariance matrix are the only fixed point states of a non-trivial beam splitter transformation. In constrast with other characterizations of Gaussian states such as the one from Hudson\citep{Hu74}, the DS theorem does not require any constraint on the purity of the state. The stability result of Theorem~\ref{teo:pids} provides an explicit estimate of the robustness of a characterization of Gaussian states through linear independence. In particular, we have obtained an estimate of the constants which reflect the fact that the quantum DS theorem is unstable when the beam-splitter is close to being transparent ($\theta=0$) or a mirror ($\theta=\pi/2$). The exact dependence on the transmitivity constant is shown in Fig.~\ref{fig:theta}. Throughout this work, we have made an effort to present explicit constants as well as to improve the order of the error parameter; however this does not mean that they are anywhere close to optimal. In fact, it is not known to us if the optimal constant must necessarily depend on the number of modes $n$ or whether the $\log(1/\varepsilon)^{-1/2}$ dependence can be lifted to a polynomial dependence. \\

The Darmois-Skitovich theorem is not only interesting as a neat characterization problem, but also because of its practical applications: it is the main theoretical concept behind the signal reconstruction method known as \textit{blind source separation}\cite{Co94} which is actively studied in the field of communication and signal processing. We hope with this study of the stability of the quantum DS theorem to stimulate a further investigation of this theorem and its extensions in the quantum information community.

%%%%%%%%%%%%%%%%%%%%%%%%%%%%%%%%%%%%%%%%%%%%%%%%%%%%%%%%%%%%%%%%%%%%%%%%%%%%%%%%%%%%%%%%%%%%%%%%%%%%%%%%%%%%%%%%%%%%%%%%%%%%%%%%%%%%%%%%%%%%%%%%%%%
%%%%%%%%%%%%%%%%%%%%%%%%%%%%%%%%%%%%%%%%%%%%%%%%%%%%%%%%%%%%%%%%%%%%%%%%%%%%%%%%%%%%%%%%%%%%%%%%%%%%%%%%%%%%%%%%%%%%%%%%%%%%%%%%%%%%%%%%%%%%%%%%%%%
%%%%%%%%%%%%%%%%%%%%%%%%%%%%%%%%%%%%%%%%%%%%%%%%%%%%%%%%%%%%%%%%%%%%%%%%%%%%%%%%%%%%%%%%%%%%%%%%%%%%%%%%%%%%%%%%%%%%%%%%%%%%%%%%%%%%%%%%%%%%%%%%%%%
%%%%%%%%%%%%%%%%%%%%%%%%%%%%%%%%%%%%%%%%%%%%%%%%%%%%%%%%%%%%%%%%%%%%%%%%%%%%%%%%%%%%%%%%%%%%%%%%%%%%%%%%%%%%%%%%%%%%%%%%%%%%%%%%%%%%%%%%%%%%%%%%%%%
%%%%%%%%%%%%%%%%%%%%%%%%%%%%%%%%%%%%%%%%%%%%%%%%%%%%%%%%%%%%%%%%%%%%%%%%%%%%%%%%%%%%%%%%%%%%%%%%%%%%%%%%%%%%%%%%%%%%%%%%%%%%%%%%%%%%%%%%%%%%%%%%%%%

\subsection{Auxiliary Lemmas}\label{technicallemmassubsection}
\begin{lemma}
Let $S\in GL(4n,\mathbb{R})$ such that
\begin{equation}
 S\begin{pmatrix}{} \varGamma & 0 \\ 0 & \varGamma \end{pmatrix}S^{T}=\begin{pmatrix}{} \cdot & 0 \\ 0 & * \end{pmatrix} \quad \text{is $2n\times 2n$ block diagonal for all symmetric $\varGamma\in\mathbb{R}^{2n\times2 n}$.}
\label{eq:1}
 \end{equation}
Then it follows that:
\begin{enumerate}[(i)]
 \item $S$ is either of the form
 \begin{enumerate}[(a)]
 \item $S=\begin{pmatrix}{} A & 0 \\ 0 & D \end{pmatrix}$ or $S=\begin{pmatrix}{} 0 & B \\ C & 0 \end{pmatrix}$, 
 \item $S=\begin{pmatrix}{} A & B \\ C & D \end{pmatrix}$, 
 \end{enumerate}
 with $A,B,C$ and $D\in\mathbb{R}^{2n\times 2n}$ invertible.
 \item \begin{equation}
 S\begin{pmatrix}{} X & 0 \\ 0 & X \end{pmatrix}S^{T}=\begin{pmatrix}{} \cdot & 0 \\ 0 & * \end{pmatrix} \quad \text{is $2n\times 2n$ block diagonal for all $X\in\mathbb{R}^{2n\times 2n}$.}
\label{eq:2}
 \end{equation}
 
\end{enumerate}
\label{l:lppm}
\end{lemma}

\begin{proof}[Proof of Lemma \ref{l:lppm}]
We decompose $S=\begin{pmatrix}{} A & B \\ C & D \end{pmatrix}$ into blocks $A,B,C$ and $D\in\reals^{n\times n}$ and observe that
equation (\ref{eq:1}) is equivalent to 
\begin{equation}
 A\varGamma C^{T}+B\varGamma D^{T}=0 \quad \text{for all $\varGamma=\varGamma^{T}$.}
 \label{eq:3}
\end{equation}
Using tensor notation this equation can be written\cite{HJ07} as
\begin{align*}
& (A\otimes C+B\otimes D)|\varGamma\rangle=0 \quad\quad\text{for all $|\varGamma\rangle\in\mathbb{R}^{4n}\otimes\mathbb{R}^{4n}$ symmetric}.\\
\iff& (A\otimes C+B\otimes D)P_{+}|X\rangle=0 \quad\quad\text{for all $|X\rangle\in\mathbb{R}^{4n}\otimes\mathbb{R}^{4n}$.}\\
\iff& (A\otimes C+B\otimes D)P_{+}=0.
\end{align*}

\begin{equation}
\iff P_{+}(A^{T}\otimes C^{T}+B^{T}\otimes D^{T})=0,
\label{eq:sym}
\end{equation}
where $P_{+}$ denotes the projector onto the symmetric subspace of $\mathbb{R}^{4n}\otimes\mathbb{R}^{4n}$.
We make the following remark that will be used frequently during this proof:
the symmetrization or symmetric component of a non-zero product state does not vanish. Suppose it does, then $P_{+}(|v\rangle\otimes|w\rangle)=0$ and 
$|v\rangle\otimes|w\rangle+|w\rangle\otimes|v\rangle=0$. Performing the scalar product with $\langle v|\otimes\langle w|$ 
leads to $\norm{v}^{2}\norm{w}^{2}+|\langle v|w\rangle|^{2}=0$ which is only zero if and only if $|v\rangle=|w\rangle=0$.\\

\textbf{On (i):} We now prove the first part of the Lemma by considering the two cases:\\
(a) One of the submatrices $A$,$B$,$C$ or $D$ is zero: we only treat the case $A=0$, the others are similar.
Then $A^{T}\otimes C^{T}=0$ and $B^{T}$ is invertible (otherwise $S$ would not have full rank). Since any non-zero product
$B^{T}v\otimes D^{T}w$ (for some $v,w\in\mathbb{R}^{2n}$) would contain a non-vanishing symmetric component, 
equation (\ref{eq:sym}) implies that $D^{T}=0$.

(b) We prove by contradiction that in this case the four submatrices are invertible. For instance, assume that $A$ is not invertible. Then there exists
a vector $0\neq|a\rangle\in\text{Ker}A^{T}$.  We choose $|e\rangle\notin\text{Ker}D^{T}$ (recall that $D\neq0$) and with (\ref{eq:sym}) we then find

\begin{equation}
0=P_{+}(A^{T}\otimes C^{T}+B^{T}\otimes D^{T})|a\rangle\otimes|e\rangle=P_{+}(B^{T}|a\rangle\otimes D^{T}|e\rangle).
\end{equation}
By the same argument as in $(a)$ we now conclude $|a\rangle\in\text{Ker}B^{T}$. Moreover,

\begin{equation*}
S^{T}\begin{pmatrix}a\\0 \end{pmatrix} = \begin{pmatrix}A^{T} & C^{T} \\ B^{T} & D^{T} \end{pmatrix} \begin{pmatrix}a\\0 \end{pmatrix}=0.
\end{equation*}

The latter is in contradiction to the invertibility of $S$ (if $S$ is an invertible matrix then the kernel is trivial),
hence A --and due to analogous reasoning-- $B,C$ and $D$ are invertible.\\

\textbf{On (ii):} We now proof the second part of the Lemma, Eq. (\ref{eq:2}), by showing 
that the equivalent expression $AX C^{T}+BX D^{T}=0$ 
for all $X\in\mathbb{R}^{2n\times 2n}$ is true.\\
This is trivially satisfied if $S$ is given in form of case (i,a). Therefore we are left with the case (i,b) where in particular $B$ and $C$ are invertible.
W.l.o.g. we choose $B=C=\iden$ (this can be done by redefining $A\to A^{-1}B$ and $D\to D^{-1}C$ in (\ref{eq:3})) 
and equation (\ref{eq:sym}) reads

\begin{equation}
  P_{+}(A^{T}\otimes \iden+\iden\otimes D^{T})=0.
  \label{eq:44}
\end{equation}
We now show in three steps that $A^{T}\otimes \iden+\iden\otimes D^{T}=0$, which then concludes the proof.
First we show that there exists $\lambda\in\mathbb{C}$ such that $Spec(A^{T})=\{\lambda\}$ and $Spec(D^{T})=\{-\lambda\}$.
To this purpose we choose 
eigenvectors $|e\rangle$ of $A^{T}$ and $|f\rangle$ of and $D^{T}$ with eigenvalues $\lambda$ and $\omega$ respectively. Then
\begin{equation*}
 0=P_{+}(A^{T}\otimes\iden+\iden\otimes D^{T})|e\rangle\otimes|f\rangle=(\lambda+\omega)P_{+}(|e\rangle\otimes|f\rangle).
\end{equation*}
Again, since the symmetrization of a non-zero product state is different from zero, we find $\lambda=-\omega$. Note that this holds for arbitrary
eigenvalues $\lambda$ of $A^{T}$ and $\omega$ of $D^{T}$. \\
Second, using the Jordan normal form decomposition, we decompose $A^{T}$ (and $D^{T}$) into a diagonalizable $\lambda\iden$ and nilpotent part $N_{A} (N_{D})$
and observe that $(A^{T}\otimes\iden+\iden\otimes D^{T})=(\lambda\iden\otimes\iden+N_{A}\otimes\iden)+\iden\otimes(-\lambda\iden+N_{D})=(N_{A}\otimes\iden+\iden\otimes N_{D})$.\\
Finally, equation (\ref{eq:44}) reads

\begin{equation}
P_{+}(N_{A}\otimes\iden+\iden\otimes N_{D})=0, 
\label{eq:55}
\end{equation}
and we can conclude the proof by deriving that this implies $(N_{A}\otimes\iden+\iden\otimes N_{D})=0$. This is the third step. \\

Assume $N_{A}\otimes\iden+\iden\otimes N_{D}\neq0$. Using the symmetry argument about non-zero product states we find 
$N_{A}\neq0$ and $N_{D}\neq0$. 
Let $s$ be such that $N_{A}^{s}=0$ and $N_{A}^{s-1}\neq0$. Then we multiply (\ref{eq:55}) from the right by $N_{A}^{s-1}\otimes\iden$ to get $P_{+}(N_{A}^{s-1}\otimes N_{D})=0$.
 But this, in turn, implies $N_{A}^{s}\otimes N_{D}=0$ and leads to a contradiction. Therefore 
\begin{equation}
 (N_{A}\otimes\iden+\iden\otimes N_{D})=0.
\end{equation}

\end{proof}

\begin{lemma}
 Let $\rho_{12}$  be a density operator of a bipartite system with reduced states $\rho_{1}$ and $\rho_2$. If $\tilde{\rho_{1}}\otimes\tilde{\rho_{2}}$ describe an arbitrary product state and $\norm{\rho_{12}-\tilde{\rho_{1}}\otimes\tilde{\rho_{2}}}_{1}\leq\varepsilon$, then
 $\norm{\rho_{12}-\rho_{1}\otimes\rho_{2}}_{1}\leq3\varepsilon$.
\label{lemma:closetraces}
 \end{lemma}
\begin{proof}
Using the triangle inequality twice, we find
\begin{align*}
 \norm{\rho_{12}-\rho_{1}\otimes\rho_{2}}_{1} &\leq \norm{\rho_{12}-\tilde{\rho_{1}}\otimes\tilde{\rho_{2}}}_{1}+\norm{\tilde{\rho_{1}}\otimes\tilde{\rho_{2}}-\rho_{1}\otimes\rho_{2}}_{1},\\
 &\leq \varepsilon + \norm{\tilde{\rho_{1}}\otimes\tilde{\rho_{2}}-\tilde{\rho_{1}}\otimes\rho_{2}}_{1}+ \norm{\tilde{\rho_{1}}\otimes\rho_{2}-\rho_{1}\otimes\rho_{2}}_{1},\\
 &=\varepsilon+\norm{\tilde{\rho_{2}}-\rho_{2}}_{1}+\norm{\tilde{\rho_{1}}-\rho_{1}}_{1}.
\end{align*}
Exploiting  that $ \norm{X}_{1}=\sup_{Y:\norm{Y}\leq1}\ \Tr[YX]$ we can bound
\begin{align*}
 \norm{\tilde{\rho_{1}}-\rho_{1}}_1&=\sup_{Y:\norm{Y}\leq1}\ \Tr[(Y\otimes\iden) (\tilde{\rho_{1}}\otimes\tilde{\rho_{2}}-\rho_{12})],\\
 &\leq \sup_{\hat{Y}:\norm{\hat{Y}}\leq1}\ \Tr[\hat{Y} (\tilde{\rho_{1}}\otimes\tilde{\rho_{2}}-\rho_{12})],\\
 &=\norm{\tilde{\rho_{1}}\otimes\tilde{\rho_{2}}-\rho_{12}}_{1}\leq\epsilon,
\end{align*}
and similar for the other term.
\end{proof}

\begin{proof}[Proof of Lemma (\ref{lemma:region})]

We follow the proof idea of Lemma 1 from Ref. \onlinecite{Ga74}.
% Let \[ \alpha:=\max\left(\frac{1}{\cos\theta},\frac{1}{\sin\theta}\right)\geq1. \]

W.l.o.g assume $0<\theta\leq\pi/4$ and set $\eta_{2}=\tan\theta\eta_{1}, \eta_{1}=\xi$ in Eq.~\ref{eq:pfegs}. In case $\theta>\pi/4$, set $\eta_{1}=\eta_{2}/\tan\theta$ in Eq.~\ref{eq:pfegs} and proceed likewise. Then Eq.~\ref{eq:pfegs} becomes 
\begin{equation*}
 \chi_{1}\left(\left(1+\tan^{2}\theta\right)\cos\theta\xi\right)=
 \chi_{1}\left(\cos\theta\xi\right)\chi_{1}(\sin\theta\tan\theta\xi)\chi_{2}(\cos\theta\tan\theta\xi)\chi_{2}(-\sin\theta\xi)
+G(\xi,\tan\theta\xi),
\end{equation*}
for all $\xi\in\reals^{2n}$. Replace $\xi\mapsto(\xi/\cos\theta)$ in the previous equation to obtain 
\begin{equation*}
 \chi_{1}\left(\left(1+\tan^{2}\theta\right)\xi\right)=
 \chi_{1}\left(\xi\right)\chi_{1}(\tan^{2}\theta\xi)|\chi_{2}(\tan\theta\xi)|^{2}
+G\left(\frac{\xi}{\cos\theta},\frac{\tan\theta\xi}{\cos\theta}\right).
\end{equation*}

Since $\norm{G}\leq3\varepsilon$, we have for all $\xi\in\reals^{2n}$:
\begin{equation*}
 \left|\chi_{1}\left(\left(1+\tan^{2}\theta\right)\xi\right)\right|\geq 
 |\chi_{1}(\xi)||\chi_{1}(\tan^{2}\theta\xi)||\chi_{2}(\tan\theta\xi)|^{2}
-3\varepsilon.
\end{equation*}

Similarly, for $0<\theta\leq\pi/4$ and $\eta_{1}=-\tan\theta\eta_{2},\eta_{2}=\xi$ we arrive to 
\begin{equation*}
 \left|\chi_{2}\left(\left(1+\tan^{2}\theta\right)\xi\right)\right|\geq 
 |\chi_{2}(\xi)||\chi_{2}(\tan^{2}\theta\xi)||\chi_{1}(\tan\theta\xi)|^{2}
-3\varepsilon.
\end{equation*}

With $\gamma(\xi):=\min_j\min_{\norm{\eta}<\norm{\xi}}\left|\chi_{j}(\eta)\right|$ we obtain 
$
 \gamma\left(\left(1+\tan^{2}\theta\right)\xi\right)\geq \gamma^{4}\left(\xi\right)-3\varepsilon
$.
Replacing $\xi$ by $\left(1+\tan^{2}\theta\right)^{k}\xi$ with $k\in\mathbb{N}$ in the previous equation gives
\begin{equation}
 \gamma\left(\left(1+\tan^{2}\theta\right)^{k+1}\xi\right)\geq \gamma^{4}\left(\left(1+\tan^{2}\theta\right)^{k}\xi\right)-3\varepsilon \qquad \text{for all }\xi\in\reals^{2n},k\in\mathbb{N}.
\label{eq:rproof1}
 \end{equation}

It is a fact that for any classical characteristic function $\phi(t)$ with variance $\lambda$ the following inequality holds (see for instance Ref. \onlinecite[p. 89]{U99})

\begin{equation}
 |\phi(t)|\geq 1- \frac{1}{2}\lambda t^{2}.
 \label{eq:cineq}
\end{equation}

Let us fix $\xi$ in the direction of phase space in which we obtain the largest variance $\lambda$ of $\rho_{1}$ and $\rho_{2}$
and consider the region where $\norm{\xi}_{2}\leq\sqrt{\frac{1}{\lambda}}$. Then from Eq. \ref{eq:cineq} for $\norm{\xi}_{2}\leq\sqrt{\frac{1}{\lambda}}$ we have $|\chi(\xi)|\geq1/2$.

Using Eq. (\ref{eq:rproof1}) and the inequality $(1-a)^{n}\geq1-na, \forall n\in\mathbb{N}, \forall a\in[0,1]$, we can show by induction that

\begin{equation*}
 \gamma\left(\left(1+\tan^{2}\theta\right)^{k+1}\xi\right)\geq \left(\frac{1}{2}\right)^{4^{k}}-4\varepsilon \qquad \text{for }k\in\mathbb{N},\norm{\xi}_{2}\leq\sqrt{\frac{1}{\lambda}}.
\end{equation*}
Moreover, for $\varepsilon< 1$ we clearly have 

\begin{equation}\label{eq:induction}
 \gamma\left(\left(1+\tan^{2}\theta\right)^{k+1}\xi\right)\geq \left(\frac{1}{2}\right)^{4^{k}}-4\varepsilon^{1/12} \qquad \text{for }k\in\mathbb{N},\norm{\xi}_{2}\leq\sqrt{\frac{1}{ \lambda}}.
\end{equation}

Finally, we take $k_{0}$ such that $2^{k_{0}}=\sqrt{\log_{2}\frac{1}{\varepsilon^{1/12}}}$, to obtain with the help of Eq.~\eqref{eq:induction}
\begin{align*}
\gamma\left(\left(1+\tan^{2}\theta\right)^{k_{0}}\xi\right)&\geq\left(\frac{1}{2}\right)^{4^{k_{0}-1}}-4\varepsilon^{1/12},\\
&= 12\varepsilon^{1/12}.
\end{align*}
We thus have $\gamma\left(\left(1+\tan^{2}\theta\right)^{k_{0}}\xi\right)>\varepsilon^{1/12}$ and $\left(1+\tan^{2}\theta\right)^{k_{0}}\xi\in\mathfrak{B}_{r}$ as claimed.
\end{proof}

\bigskip
\emph{Acknowledgments:}
I would like to thank Michael M. Wolf for his continuous interest and many insightful discussions thoughout this work. In particular, I thank him for suggesting the problem considered in this paper. In addition, I would like to thank A. Michelangelo for an insightful discussion of Lemma 1 from Ref.~\onlinecite{Ga74}, and M. Keyl for clarifications concerning Schwarz operators. Furthermore, I wish to thank M. Christandl and J. P. Solovej for their hospitality and the financial support received during my visit at the QMATH Center in Copenhagen.  This work was partly written during this visit.

\newpage

\section{Appendix}\label{sec:appendix}

\subsection{Upper bound of $\norm{V}_{2}$, Eq. \ref{eq:bddV}.}\label{appen:A}

Let us write $R_{1k}, k=1,\ldots2n,$ and $R_{2l}, l=1,\ldots2n,$ for the entries of the vector $R_{1}=(Q_{1},P_{1},\ldots,Q_{n},P_{n})$ and $R_{2}=(Q_{n+1},P_{n+1},\ldots,Q_{2n},P_{2n})$ respectively. First, from the orthogonality of the symplectic matrix $\sigma$

\begin{align*}
\norm{V}_{2}&=\frac{\norm{\sigma(\Tr gR_{1}R_{2}^{T})\sigma^{T}}_{2}}{|\tan\theta|}\\
&=\frac{\norm{\Tr gR_{1}R_{2}^{T}}_{2}}{|\tan\theta|}=\frac{\left(\sum_{kl}|\Tr gR_{1k}R_{2l}|^{2}\right)^{1/2}}{|\tan\theta|}.
\end{align*}  
Our plan is to bound each entry $|\Tr gR_{1k}R_{2l}|^{2}$. Since the operator $g$ is the difference of two Schwartz operators it is also Schwartz. Thus we have from Theorem \ref{theo:schwprop} $(iii)-(iv)$ that $|g|^{1/2}$ is a Schwartz operator. If $g=\sum \lambda_{i}|i\rangle\langle i |$ is the spectral decomposition of $g$, we consider the following factorization 
 
 \begin{equation*}
 g=|g|Q, \qquad \text{with} \qquad Q:=\sum_{i}\sgn(\lambda_{i})|i\rangle\langle i|.
 \end{equation*}
 
Clearly, $Q=Q^{-1}$ commutes with $|g|$ and $|g|^{1/2}$.  
 Moreover, from Theorem \ref{theo:schwprop} $(i)$ we will be able to use the trace cyclicity for our following computation. Using twice the Cauchy-Schwarz inequality and $\norm{g}_{1}\leq 3\varepsilon$, we find  

\begin{align*}
|\Tr gR_{1k}R_{2l}|^{2}=|\Tr |g|^{1/2}|g|^{1/2}Q R_{1k}R_{2l}|^{2} & \leq \left(\Tr |g|\right)\left(\Tr |g|R^{2}_{1k}R^{2}_{2l}\right), \\
&\leq 3\varepsilon\sqrt{\Tr|g|R^{4}_{1k}\Tr|g|R^{4}_{2l}}, \\
&\leq (3\varepsilon) \max\{\Tr|g|R^{4}_{1k},\Tr|g|R^{4}_{2l}\}.
\end{align*}

Using again the decomposition of $g$ we obtain for $j=1,2$ that
\begin{align*}
|\Tr|g|R^{4}_{jk}|=|\Tr QgR^{4}|\leq \norm{Q}\norm{gR^{4}_{jk}}_{1}&=\norm{(\rho_{ab}-\rho_{a}\otimes\rho_{b})R^{4}_{jk}}_{1}\\
&\leq 2\norm{\rho_{ab}R^{4}_{jk}}_{1},
\end{align*}
since $R^{4}_{jk}$ is a local operator on one part of the output. Consequently, $|\Tr gR_{1k}R_{2l}|^{2}\leq 6\varepsilon\max\{\Tr\rho_{ab}R^{4}_{1k},\Tr\rho_{ab}R^{4}_{2l}\}\leq 6\varepsilon \kappa$ where $\kappa:= \max\left\{ \norm{\rho_{ab}R^{2}_{\xi}R^{2}_{\eta}}_{1}\;\big|\;\norm{\xi}_{2}=\norm{\eta}_{2}=1\right\}$ is the largest generalized fourth moment of $\rho_{ab}$. Note that since $\rho_{ab}$ is a Schwartz operator $\kappa<\infty$. Thus
\begin{equation*}
\norm{V}_{2}\leq \frac{\sqrt{24 n^{2}\kappa\varepsilon}}{|\tan\theta|}.
\end{equation*}

\subsection{Upper bound of $|F(\xi)|$, Eq.\;\ref{eq:bdddF}.}

The line integral of a matrix $A\in\mathbb{C}^{2n\times 2n}$ is defined in terms of the line integrals of the rows of $A$. Namely, if $A_{k},k=1,\ldots,2n$ are the rows of $A$, then
\begin{equation*}
\int A\cdot d\eta := \begin{pmatrix}
\int A_{1}\cdot d\eta  \\
\vdots \\
\int A_{2n}\cdot d\eta
\end{pmatrix}.
\end{equation*}

Thus the line integral of a matrix-valued function is a vector, and the line integral of a vector field is a scalar.
Let us denote by $M:\reals^{2n}\to \mathbb{C}^{2n\times2n}$ a matrix-valued function. The upper bound for Eq.\;\ref{eq:bdddF} is equivalent to bound 
\begin{equation*}
F(\xi):= \int_{\mathcal{C}(\xi)}\left( \int_{\mathcal{C}(\eta)} M(z)\cdot dz\right)\cdot d\eta,
\end{equation*}
with $\xi,\eta\in\mathfrak{B}_{r/2}$ and
\begin{equation*}
M(z)=\cos^{2}\theta\left(Q_{22}(z)+\frac{Q_{12}(z)}{\tan\theta}\right).
\end{equation*}
Let us parametrize the curve $\mathcal{C}(\xi)$ via $\vartheta:[0,1]\to \mathcal{C}$, $t\mapsto t\xi$ and write $M_{kl}(z)$ for the entries of the matrix $M(z)\in\comp^{2n\times 2n}$. Define the matrix-valued function $\reals^{2n}\ni z\mapsto Y(z)\in\comp^{2n\times 2n}$ to have the entries $Y_{kl}(z):=\int_{0}^{1}M_{kl}(sz)ds$. From the explicit parametrization of the line integrals and the Cauchy-Schwarz inequality we obtain that
\begin{align}
|F(\xi)|&\leq  \max_{t\in[0,1]} \left|t\xi\cdot Y(t\xi)\xi \right|, \nonumber \\
&\leq \norm{\xi}^{2}_{2} \; \max_{t\in[0,1]} \left(\sum_{k,l=1}^{2n}|Y(t\xi)_{kl}|^{2} \right)^{1/2} \label{eq:boundF}.
\end{align}

In order to bound $|Y(t\xi)_{kl}|$ we differentiate Eq.~\eqref{eq:q} with the help of Lemma~\ref{l:qcfmoments2} to find 
\begin{equation*}
  Q_{22}(z)=\frac{G_{2}(z)G^{T}_{2}(z)}{\chi^{2}_{1}(\cos\theta z)\chi^{2}_{2}(-\sin\theta z)}-\frac{G_{22}(z)}{\chi_{1}(\cos\theta z)\chi_{2}(-\sin\theta z)}\in \comp^{2n\times 2n},
\end{equation*}
\begin{equation*}
Q_{12}(z)=-\frac{G_{12}(z)}{\chi_{1}(\cos\theta z)\chi_{2}(-\sin\theta z)}-\frac{\nabla\chi_{\rho_{a}}(z)G_{2}^{T}(z)}{\chi^{2}_{1}(\cos\theta z)\chi^{2}_{2}(-\sin\theta z)} \in \comp^{2n\times 2n},
\end{equation*}

for $z\in\reals^{2n}, \norm{z}_{2}\leq r/2$ where
\begin{align*}
G_{2}(z)&=-\frac{i}{2}\sigma \Tr[e^{iz\cdot\sigma R_{1}}\{R_{2},g\}]\in  \comp^{2n},\\
G_{22}(z)&=-\frac{1}{4}\sigma\Tr[e^{iz\cdot\sigma R_{1}}\{\{R_{2},g\},R^{T}_{2}\}]\sigma^{T} \in \comp^{2n\times 2n},\\
G_{12}(z)&=-\frac{1}{4}\sigma\Tr[e^{iz\cdot\sigma R_{1}}\{\{R_{1},g\},R^{T}_{2}\}] \in \comp^{2n\times 2n},\\
\nabla\chi_{\rho_{a}}(z)&=-\frac{i}{2}\Tr[e^{iz\cdot\sigma R_{1}}\{R_{1},\rho_{a}\}]\in  \comp^{2n}.
\end{align*}

Let us write $R_{1k}, k=1,\ldots2n,$ and $R_{2l}, l=1,\ldots2n,$ for the entries of the vector $R_{1}=(Q_{1},P_{1},\ldots,Q_{n},P_{n})$ and $R_{2}=(Q_{n+1},P_{n+1},\ldots,Q_{2n},P_{2n})$ respectively. From Lemma~\ref{lemma:region} we know that for $z\in \mathfrak{B}_{r}$, $\chi(z)>12 \varepsilon^{1/12}$ and therefore we can upper bound each entry of $Y(t\xi)_{kl}$ by

\begin{align*}
|Y(t\xi)_{kl}| &\leq \cos^{2}\theta \left( \frac{\norm{\{\{R_{2k},g\},R_{2l}\}}_{1}}{4(12\varepsilon^{1/12})^{2}}+\frac{\norm{\{R_{2k},g\}}_{1}\norm{\{R_{2l},g\}}_{1}}{4(12\varepsilon^{1/12})^{4}} +  \frac{\norm{\{\{R_{1k},g\},R_{2l}\}}_{1}}{4(12\varepsilon^{1/12})^{2}\tan\theta} +\frac{\norm{\{R_{1k},\rho_{a}\}}_{1}\norm{\{R_{2l},g\}}_{1}}{4(12\varepsilon^{1/12})^{4}\tan\theta} \right). \\
\end{align*}

Following a similar procedure as for the bound of $\norm{V}_{2}$ (see Apendix \ref{appen:A}) we obtain 

\begin{align*}
|Y(t\xi)_{kl}|^{2} &\leq \left(\frac{\kappa\cos^{2}\theta}{8\tan^{2}\theta}\right)\varepsilon^{2/3}, \\
\end{align*}
so from Eq.~\eqref{eq:boundF}

\begin{equation*}
\left|F\left(\frac{\xi}{\cos\theta}\right)\right|^{2}\leq \left(\frac{n^{2}\kappa \norm{\xi}_{2}^{4}}{2\tan^{2}\theta}\right)\varepsilon^{2/3}.
\end{equation*}

% \bibliographystyle{natbib}
%\bibliography{GBSbibliography}

\end{document}